\documentclass[journal,comsoc]{IEEEtran}
\usepackage[T1]{fontenc}
\usepackage[latin9]{inputenc}
\usepackage[english]{babel} 
\usepackage{color}
\usepackage{units}
\usepackage{amsmath}
\usepackage{amsthm}
\usepackage{graphicx}

\makeatletter
  \theoremstyle{plain}
  \newtheorem{lem}{\protect\lemmaname}
  \theoremstyle{plain}
  \newtheorem{thm}{\protect\theoremname}
 \ifx\proof\undefined\
   \newenvironment{proof}[1][\proofname]{\par
     \normalfont\topsep6\p@\@plus6\p@\relax
     \trivlist
     \itemindent\parindent
     \item[\hskip\labelsep
           \scshape
       #1]\ignorespaces
   }{%
     \endtrivlist\@endpefalse
   }
   \providecommand{\proofname}{Proof}
 \fi

\usepackage{amsmath,amssymb,amsthm}
\usepackage{comment}

\usepackage{graphicx}
\usepackage{verbatimbox}
\usepackage{listings}
\usepackage{cite}
\usepackage{blindtext}
\usepackage{graphicx}
\usepackage{subfigure}

\usepackage[linesnumbered,ruled]{algorithm2e}
\usepackage[noend]{algpseudocode}

\allowdisplaybreaks
\usepackage{empheq}
\usepackage{amsmath}
\usepackage{mdframed}
\mdfsetup{
linewidth=0.75pt, 
innertopmargin=0.5em, 
skipbelow=0.25em,
skipabove=0.25em,
nobreak=false
}

\usepackage{blindtext,letltxmacro,xcolor,xparse}
\usepackage{lipsum}

\newcommand*{\QEDA}{\hfill\ensuremath{\blacksquare}}%

\usepackage[colorlinks,citecolor=magenta,urlcolor=blue,bookmarks=false,linkcolor=blue,hypertexnames=true]{hyperref}

\author{
Omar~Alhussein,~\IEEEmembership{Student~Member,~IEEE,}                
and~Weihua~Zhuang,~\IEEEmembership{Fellow,~IEEE}
\thanks{
Omar~Alhussein~and~Weihua~Zhuang are with the Department of Electrical and Computer Engineering, University of Waterloo, Waterloo, ON, Canada, N2L 3G1 (emails: \{oalhusse,~wzhuang\}@uwaterloo.ca)
}
}

\IEEEpubid{\copyright~2020 IEEE}

\makeatother

\usepackage{babel}
\providecommand{\lemmaname}{Lemma}
\providecommand{\theoremname}{Theorem}

\begin{document}

\title{Robust Online Composition, Routing and NF Placement for NFV-enabled
Services}

\maketitle
\begin{abstract}
Network function virtualization (NFV) fosters innovation in the networking
field and reduces the complexity involved in managing modern-day conventional
networks. Via NFV, the provisioning of a network service becomes more
agile, whereby virtual network functions can be instantiated on commodity
servers and data centers on demand. Network functions can be either
mandatory or best-effort. The former type is strictly necessary for
the correctness of a network service, whereas the latter is preferrable
yet not necessary. In this paper, we study the online provisioning
of NFV-enabled network services. We consider both unicast and multicast
NFV-enabled services with multiple mandatory and best-effort NF instances.
We propose a primal-dual based online approximation algorithm that
allocates both processing and transmission resources to maximize a
profit function, subject to resource constraints on physical links
and NFV nodes. The online algorithm resembles a joint admission mechanism
and an online composition, routing and NF placement framework. The
online algorithm is derived from an offline formulation through a
primal-dual based analysis. Such analysis offers direct insights and
a fundamental understanding on the nature of the profit-maximization
problem for NFV-enabled services with multiple resource types.
\end{abstract}

\begin{IEEEkeywords}
NFV, online algorithms, primal-dual scheme, profit maximization, competitive
analysis.
\end{IEEEkeywords}

\section{Introduction}
Network function virtualization (NFV) has established itself as a
prominent concept for the provisioning of modern communication networks.
Traditionally, network elements, such as routers and middlewares,
were implemented in proprietary hardware boxes. With NFV, such network
elements are virtualized as NF instances, and can be deployed in the
data plane in commodity servers and cloud data centers. This gives
rise to service-customized networks \textendash{} a provisioning mechanism
that can meet new demands and agile use cases. A service-customized
network resembles a network application, whereby a data flow passes
by virtual NF instances for processing before arriving at the destination(s).
In such paradigm, a service provider submits a request to reserve
network resources (e.g., transmission and processing) to orchestrate
its own virtual network. The reserved resources should be guaranteed
according to some agreed-upon quality of service. The infrastructure
provider aims at maximizing some ``profit'' function, while minimizing
the provisioning costs of the network services.

A considerable number of works are carried out for the orchestration
and provisioning of NFV-enabled service requests \cite{Liu2019,8806628,8761451,Huang2017,Ghaznavi2017,Palkar:2015:EFN:2815400.2815423,Xie2016,Herrera}.
Earlier research focuses on the orchestration of a single service
request, where the focus is on minimizing the provisioning cost of
a single service without taking other services into consideration
\cite{Alhussein,Ghaznavi2017,Ye2016,Li}. However, the admission and
embedding of one service request affects the service provisioning
of other requests, thereby the need for orchestrating multiple service
requests jointly. To this end, the relevant literature can be classified
based on their handling of service requests into offline and online.
In the former, all service requests are known a priori, and all service
requests are assumed to arrive in one batch. In practice, network
services arrive in an online and random manner without knowledge of
future requests \cite{8482326,Lukovszki2015}.

Due to the increased flexibility and agility brought-forth by NFV,
future (over-the-top) service requests are envisaged to be hardly
predictable \cite{EVEN2013184}. Future service quality and data
traffic patterns for new use cases are arguably not well understood,
and an advanced knowledge of future patterns can be difficult to obtain
or predict. Moreover, such traffic patterns can vary dramatically
over short periods due to the inherent agility of NFV-based networks.
\IEEEpubidadjcol

Some relevant studies deal with the online handling of service requests
without statistical assumptions \cite{8482326,Devanur:2019:NOO:3299993.3284177,Huang2018,Lukovszki2015,Xu2017a,XU201815}.
The NFV-enabled frameworks are based on the seminal work by Awerbuch
et al. \cite{Awerbuch1993}, where some new aspects are due to the
inclusion of NF instances and the need for an admission mechanism
for service requests with multiple resource types. To our knowledge,
in the existing NFV-enabled works, service requests have either one
resource type or one NF instance. Also, online (routing and NF placement)
algorithms can be classified as either \textit{all-or-nothing} or
\textit{all-or-something}. In the all-or-nothing scenario, service
requests need to be fully served in the network substrate. In the
all-or-something scenario, services can be partially (fractionally)
served, e.g., admitting a service request while reducing the required
data rate \cite{EVEN2013184,TCS-024}. Current works in the relevant
NFV literature can be considered as all-or-nothing schemes.

This paper deals with two resource types simultaneously, namely the
processing and transmission resources. The two resource types are
often conflicting in their utilization. Therefore, there is a need
to design a generalized admission mechanism and an online joint composition,
routing, and NF placement (JCRP) algorithm that take the multiple
resource types into account. We consider unicast and multicast service
requests that can have multiple NF instances. Furthermore, we consider
two NF types, namely best-effort and mandatory. A successful placement
of a network service is contingent only on successfully placing the
set of mandatory NFs. The functionality of a best-effort NF is not
necessary for the correctness of a network service \cite{Choi2003}.
Therefore, the set of best-effort NFs can be removed from a service
request when it is deemed ``too prohibitive''. In practice, best-effort
NFs can improve either the performance, the quality of service, or
the security of a network service, such as in the case of compression
and intrusion detection. Consider for instance a video/image compression
NF type for a voice over IP network service. Such NF type enhances
the quality of service by compressing the incoming data flow. However,
when the available processing resources (or the available subscription)
for the NF type in the network substrate is scarce, a network service
can take a rather unnecessarily long (i.e., expensive) route, which
would be too costly and can conversely degrade the overall quality
of service. Therefore, such NF type can be declared as best-effort,
whereby including it should be contingent on whether a certain profit
is achieved.

The objective of this work is to develop a robust admission mechanism
and an online joint composition, routing and NF placement (JCRP) framework
(online algorithm, in short) that aims to maximize a profit function,
which is proportional to the so-called amortized throughput, while
considering unicast and multicast NFV-enabled services with best-effort
and mandatory NF instances, subject to resource constraints on physical
links and NFV nodes. The amortized throughput is defined as the weighted
total transmission and processing resources reserved for all the accepted
service requests.

The online algorithm relies on two main components, i.e., (i) an admission
mechanism that rejects or accepts a service request based on a profit
function while taking best-effort NFs into consideration, and (ii)
an online JCRP algorithm that provides unicast-enabled and multicast-enabled
routing and NF placement configurations for the admitted service requests.
 The online algorithm is developed through a primal-dual analysis,
which provides an approximately optimal result with provable competitive
performance. A formal defintion of the competitive ratio (performance)
is stated as follows. For a profit-maximization problem, let $A_{\mathrm{OPT}}(\sigma)$
be the profit of the (optimal) offline solution for a sequence of
requests ($\sigma$). An online algorithm is $c$-competitive if the
produced solution is feasible and its profit is at least $\nicefrac{A_{\mathrm{OPT}}(\sigma)}{c}-e$,
where $e$ is an additive term that is independent of the service
requests \cite{TCS-024}. The primal-dual approach exists for solving
offline optimization problems. Buchbinder and Naor extended the framework
for the treatment of online algorithms \cite{TCS-024}. This work
offers the following new contributions:
\begin{itemize}
\item We propose a primal-dual based online algorithm to allocate both processing
and transmission resources for network services with multiple NF instances.
In addition, we consider heterogeneous NFV-enabled services with mandatory
and best-effort NFs. We provide a natural generalization to relevant
works that focus on the provisioning of services without a processing
requirement or with only one NF instance. The online algorithm can
be regarded as an all-or-nothing/all-or-something algorithm in the
sense that the requested data rate and the required processing resource
for each NF instance should be fully satisfied, yet the set of best-effort NF
instances need not necessarily be included in the accepted service
request, thereby providing the flexibility to recompose the logical
topology of a service request before admission;
\item The primal-dual analysis offers an alternative analysis and generalized
treatment to approaches adopted in recent relevant works \cite{8482326,Huang2018,Lukovszki2015,Xu2017a,XU201815}.
For instance, the competitive performance in the aforementioned works
is shown to associate not only with an optimal integer solution but
also with an optimal fractional solution;
\item We propose a ``one-step'' algorithm for the routing and NF placement
of unicast and multicast services for an unconstrained scenario. The
algorithm relies mainly on the construction of an auxiliary network
transformation that has a one-to-one mapping from the NF placement
and routing problem to an equivalent routing problem.
\end{itemize}
The rest of the paper is organized as follows. Section \ref{sec:related_works}
gives an overview of related works. Section \ref{sec:system_model}
describes the system model under consideration, followed by the problem
description. Section \ref{sec:problem_formulation} presents the problem
formulation, which includes the design of a profit function, and the
primal-dual based problem formulation for the offline routing and
NF placement framework. In Section \ref{sec:the_approach}, we develop
the primal-dual based admission mechanism, followed by an analysis
of the competitive performance of the proposed admission mechanism.
Section \ref{sec:routing_solution} presents the routing and NF placement
algorithm for the proposed admission mechanism. Finally, Section \ref{sec:numerical}
presents some discussions on the proposed framework, followed by
simulation results to investigate and corroborate the competitive
performance of the proposed work.

\section{Related Works}

\label{sec:related_works}

\subsection{Routing and NF Placement}

\label{sec:related_works:static}

A growing body of literature has been evolving for the composition,
routing and NF placement for both unicast and multicast NFV-enabled
services. Most of the earlier works focus on the orchestration of
unicast services, cf. \cite{Herrera,Xie2016}. In its most basic form,
the orchestration of an NFV-enabled service poses two correlated and
conflicting subproblems, i.e., how to place (or select) the NF instances,
and how to route the traffic to traverse the NF instances. Placing
a minimal feasible number of NF instances can lead to a large link
provisioning cost; conversely, deploying more NF instances can reduce
the link provisioning cost at the expense of an increased function
provisioning cost. This tradeoff becomes more conspicuous when considering
a multicast service in which traffic is routed to more than one destination.
To that end, several works have been proposed for the NF placement
and routing of a multicast service \cite{Alhussein, 8416286,8057289,7949017, Zhang2015,Zhang2016,Zeng2016,Huang2014,Blendin2015}.
The aformentioned works consider designing heuristic (or approximation)
algorithms to orchestrate one service request without taking future
service requests into consideration. In this paper, we consider an
online setting, for which more relevant literature is addressed in
Subsection \ref{sec:related_works:online}. 

\subsection{Competitive Online Routing (Predating NFV)}

Prior to enabling NFV, in traditional circuit switching, call requests
(respectively, service requests) resembled a routing request from
the source to the destinations with a data rate requirement that need
to be routed in a capacitated network substrate \cite{Awerbuch1993}.
With the emergence of NFV, service requests subsumed call requests
with the additional requirement that virtual NFs need to be instantiated
on commodity servers or data centers (NFV nodes) along the route.

Related works can be classified according to the parameter that measures
the performance of the intended design, typically the throughput or
the congestion. In throughput-maximization frameworks, we measure
the transmission resources of all admitted service requests. In congestion-minimization,
we measure the maximum link congestion, i.e., the maximum ratio of
the allocated transmission resources on a link to its total transmission
capacity. This work can be considered as a generalized case of the
throughput-maximization framework. 

In 1993, Aspnes et al. developed a competitive strategy for congestion-minimization
that achieves a competitive ratio of $\mathcal{O}(\log n)$ for service
requests of infinite holding time, where $n$ is the total number
of nodes in the substrate network \cite{Aspnes1993}. Assuming service
requests have finite holding time (which is revealed only upon the
arrival for each service request), the authors extended their result
to achieve an $\mathcal{O}(\log nT)$ competitive ratio, where $T$
is the maximum holding time of all service requests. For the throughput-maximization
model, Awerbuch et al. achieved a competitive ratio of $\mathcal{O}(\log n)$
\cite{Awerbuch1993}. 

\subsection{Competitive Online Routing and NF Placement}

\label{sec:related_works:online}

One main challenge in devising a competitive online routing and NF
placement algorithm is pertaining to the inclusion of processing resources
along with transmission resources. To this end, Lukovszki et al.
consider that all unicast service requests contain the same set of
requested NFs and identical transmission rates, where service requests
vary with regard to the source and destination \cite{Lukovszki2015}.
They propose an $\mathcal{O}(\log K)$-competitive admission mechanism,
where $K$ is the number of NFs of a service request. Interestingly,
the competitive ratio is logarithmic with the number of NFs, which
is small in practice. In \cite{Huang2016,Huang2018}, the authors
consider both unicast and multicast requests without NFs, where routing
a request utilizes transmission resources from the physical links
and routing rules from the forwarding table of the traversed switches.
Notably, the achieved competitive ratio can be shown to be $\mathcal{O}(\max\{\log2L,\log2E\})$,
where $L$ and $E$ are the maximum number of physical links and switches
for a service request, respectively. The competitive ratio for the
two resources is balanced since $L=E-1$.

Xu et al. consider a multicast request with one NF, and develop an
$\mathcal{O}(\log L)$-competitive algorithm \cite{Xu2017a}. Ma et
al. consider the dynamic admission of delay-aware requests for services
in a distributed cloud with the objective of maximizing a designed
profit function \cite{8482326}. They first provide a heuristic algorithm
for the delay-aware scenario, followed by an online algorithm with
an $\mathcal{O}(\log L)$-competitive ratio for the special case where
the end-to-end delay is negligible.

In this paper, we propose a primal-dual based online algorithm that
accommodates both unicast and multicast service requests with multiple
NF instances that can be deployed at different NFV nodes. We consider
heterogeneous services with best-effort and mandatory NFs. In doing
so, we propose an all-or-nothing/all-or-something admission mechanism
that re-composes the logical topology of the service request before
admitting it (depending on whether or not the set of best-effort NF instances
can be included). Moreover, based on a primal-dual framework, we
offer new alternative, generalized description, and analysis to the
aforementioned NFV-enabled literature.

\section{System Model and Problem Description}

\label{sec:system_model} In this section, we present the system model
under consideration, followed by the problem description. 

\subsection{Network Functions}

With NFV, traditional applications and functionalities (which used
to be implemented in the control plane or at the end-users) are now
deployable in the data plane in NFV nodes. Examples of NFs include
firewall, intrusion detection, Web cache, proxy, and service gateway.
 From the perspective of quality of service, we consider two types
of NFs, namely mandatory and best-effort. As discussed, a successful
placement of a network service is contingent on successfully placing
only the set of mandatory NFs, whereas best-effort NFs are not necessary
for the correctness of a network service.

\subsection{Online Service Requests}

We consider an ongoing input sequence of unicast and multicast service
requests $\sigma=(S^{1},S^{2},\dots)$ that arrive in an online fashion.
The $r$th service request is expressed as
\begin{equation}
S^{r}=(s^{r},\mathcal{D}^{r},\mathcal{V}^{r},d^{r}),\qquad S^{r}\in\sigma\label{eq:service-1-1}
\end{equation}
where the source and destination nodes are $s^{r}$ and $\mathcal{D}^{r}$,
respectively; parameter $d^{r}$ denotes the required transmission
rate in packet per second (packet/s); $\mathcal{V}^{r}=\{f_{1}^{r},f_{2}^{r},\dots,f_{|\mathcal{V}^{r}|}^{r}\}$
represents the set of NFs that need to be traversed in an ascending
order for the source-destination pair. For simplicity, for each service
request, each NF requires an equal amount of processing resources
of $C(f^{r})$ in packet/s. The online algorithm to be developed
thereafter can be generalized for NFs with arbitrary processing requirements.
The sets of mandatory and best-effort NFs are denoted by $\mathcal{V}_{m}^{r}$
and $\mathcal{V}_{b}^{r}$, respectively.

\subsection{Network Substrate}

We are given a capacitated network substrate $\mathcal{G}=(\mathcal{N},\mathcal{L})$,
where $\mathcal{N}$ and $\mathcal{L}$ are the sets of nodes and
links, respectively. Each physical link $l$ ($\in\mathcal{L}$) has
a residual transmission resource, $B(l)$, in packet/s. Each node
$n$ ($\in\mathcal{N}$) has a residual processing resource, $C(n)$,
in packet/s. Nodes can be either (i) switches that are capable of
forwarding traffic only (with $C(n)=0$), or NFV nodes (e.g., commodity
servers) that are capable of both forwarding traffic and operating
a set of NF instances. An NFV node is capable of provisioning a number
of NF instances simultaneously as long as the available processing
resources satisfy the deployed NF processing requirements. Denote
the set of NFV nodes that can host an NF $f_{i}$ by $\mathcal{F}_{i}$,
where $\mathcal{F}_{i}\subseteq\mathcal{N}$.

\subsection{Problem Description}

\label{sec:problem_descriptions}

We are given a sequence of service requests $\sigma$ that is revealed
over time, i.e., the service requests arrive one by one without knowledge
of future arrivals. We need to define a profit function, whose main
goal is to maximize the amortized processing and transmission throughput.
Recall the amortized throughput is the weighted total transmission
and processing resources reserved for all the accepted service requests.
However, the profit function (and the online algorithm) should capture
both the mode of communication (i.e., unicast and multicast) and the
heterogeneity of the NF types (i.e., best-effort and mandatory). That
is, maximizing the amortized throughput alone would unfairly favor
unicast to multicast services due to the larger number of destinations
in the latter. Yet, the multicast mode of communication is more efficient
as it is shown to reduce the bandwidth consumption in backbone networks
by over 50\% in contrast to the unicast mode \cite{Malli1998}. Moreover,
although best-effort NFs are optional, their use should be incentivized. 

In the following, we first define a profit function that accurately
captures the system model and operational design requirements. Then,
we develop a path-based formulation for an offline profit-maximization
problem. The offline formulation is omniscient, where it has complete
a priori knowledge of the entire sequence of service requests. Moreover,
it yields the optimal combination of service requests and their routing
and NF placement configurations, such that the profit function is
maximized. Given the offline formulation, through a primal-dual analysis,
we develop an all-or-nothing/all-or-something online algorithm to
deal with each service request in a dynamic manner, while providing
competitive guarantees against the optimal offline adversary.

\section{Problem Formulation}

\label{sec:problem_formulation}

\subsection{The Objective (Profit) Function}

\label{sec:problem_formulation:profit}

We consider two profit functions, $\varrho^{r}$ and $\rho^{r}$,
that correspond to the transmission and processing resource types,
respectively. We know that the number of physical links required for
a service in multicast mode is always less than or equal to the overall
number of links needed for the equivalent services in unicast mode.
Therefore, for the $r$th service, an upper bound on the ratio of
the number of links in a multicast topology to the number of links
in an equivalent unicast service is given by $|\mathcal{D}^{r}|.$
Notably, it has been experimentally shown that the respective ratio
is $|\mathcal{D}^{r}|^{k}$, where $k=0.8$ for many real and generated
network topologies \cite{Chuang2001}. Therefore, for the transmission
resources, to provide a non-discriminatory treatment between multicast
and unicast services, we define $\varrho^{r}$ to be proportional
to both (i) the required data rate of the service request and (ii)
the $k$th power of the number of included destinations,
\begin{equation}
\varrho^{r}=d^{r}|\mathcal{D}^{r}|^{k},\qquad S^{r}\in\sigma\label{eq:profit1-1}
\end{equation}
where $k=0.8$. 

For the processing resources, let the amount of the incentive (and
disincentive) for including (and excluding) the set of best-effort
NFs for the $r$th service request be given by $\eta_{b}^{r}$ (and
$\eta_{m}^{r}$), respectively, where $\eta_{b}^{r}\geq\eta_{m}^{r}\geq1$.
We let $\rho^{r}$ be proportional to (i) the processing throughput
accrued from placing the NFs, and (ii) the incentive from including
(or excluding) the set of best-effort NFs,

\begin{equation}
\rho^{r}=\eta^{r}C(f^{r}),\qquad S^{r}\in\sigma\label{eq:profit2-1}
\end{equation}
where $\eta^{r}\in\{\eta_{m}^{r},\eta_{b}^{r}\}$ is a decision variable,
with $\eta^{r}=\eta_{b}^{r}$ indicating that the set of best-effort
NFs from the $r$th service is included, and $\eta^{r}=\eta_{m}^{r}$
indicating otherwise. One method is to set $\eta^{r}$
to the number of included NFs in the $r$th service (e.g., $\eta_{b}^{r}=|\mathcal{V}^{r}|$
when all best-effort NFs are accepted, and $\eta_{m}^{r}=|\mathcal{V}_{m}^{r}|$
otherwise). In contrast to existing relevant literature, $\rho^{r}$
varies with the logical topology of the composed service request.
Solutions that exclude the set of best-effort NFs can have lower provisioning
costs since they require less number of NF instances, and therefore
can be accepted if otherwise not feasible by the admission mechanism.
Since we have two resource types, the overall profit from accepting
the $r$th service request is given by $\alpha\varrho^{r}+\beta\rho^{r}$,
where $\alpha$ and $\beta$ are two coefficients to indicate the
relative importance (or scalarization) of each profit function, with
$\alpha,\beta\geq1$.

\subsection{Primal-Dual Schema}

First, we develop a path-based formulation for the offline multi-resource
profit-maximization problem. There are two possible routing and NF
placement models for the offline formulation, namely unsplittable
(or fixed) and splittable. In the unsplittable model, a service request
is restricted to an integral solution, where only one path is used
for a unicast service (or only one tree is used for a multicast service).

Formally, let all the possible paths/trees for unicast/multicast service
request $S^{r}$ be given by set $\mathcal{P}(r)$. Let $P$ ($\in\mathcal{P}(r)$)
be a path/tree on the network substrate that is selected to host service
request $S^{r}$. Here, $P$ comprises the physical links and NFV
nodes that host the virtual edges and NF instances, respectively.
Hereafter, the term ``path'' is used liberally; throughout the paper,
$P$ can be regarded as a \textit{tree} when provisioning a multicast
service. Define $y_{P}^{r}$ as the fraction of flow allocated for
service $S^{r}$ along path $P$ ($\in\mathcal{P}(r)$). In the unsplittable
model, $y_{p}^{r}\in\{0,1\}$ and $\sum_{p\in\mathcal{P}(r)}y_{p}^{r}=1$.
In the splittable model, a service can be fractionally routed on several
paths, where multipath routing is enabled and NF instances can be
split, i.e., $y_{p}^{r}\in[0,1]$ and $\sum_{p\in\mathcal{P}(r)}y_{p}^{r}=1$.
Clearly, the optimal splittable model provides a larger profit compared
to the unsplittable variant. Even stronger performance (in terms of
maximizing the profit) can be achieved when the splittable model is
linearly relaxed to allow for the sum of fractional allocations to
be at most 1, i.e., $\sum_{p\in\mathcal{P}(r)}y_{p}^{r}\leq1$. This
is due to the fact that the linear relaxation provides an upper bound
to the unsplittable (combinatorial) problem.

In data networks (such as in fifth-generation networks and the Internet),
the scale of demands is large relative to the granularity at which
it can be managed/routed \cite{Chekuri:2004:AMF:1007352.1007383},
especially in software-defined networks. Therefore, it is desired
to design and measure the (competitive) performance of a designed
apparatus against a splittable offline model (i.e., with multipath
routing and NF splitting) \cite{4031365}. Therefore, although the
online algorithm provides an unsplittable solution, its performance
will be measured against the splittable offline model.

The path-based offline profit-maximization formulation for the linearly-relaxed
splittable model is expressed in \eqref{eq:dual_fractional_LP}.

\begin{mdframed}

\center{Dual -- profit-maximization problem}

\begin{subequations}
\label{eq:dual_fractional_LP}

\begin{flalign}
& \hspace{-0.5cm} \mathrm{max}\,  \alpha\sum_{S^{r}\in\mathcal{\sigma}}\sum_{P\in\mathcal{P}(r)}\varrho^{r}y_{P}^{r}+\beta\sum_{S^{r}\in\mathcal{\sigma}}\sum_{P\in\mathcal{P}(r)}\rho^{r}y_{P}^{r}\label{eq:dual_fractional_LP:obj}\\ 
& \hspace{-0.5cm} \mathrm{subject\,to}:\nonumber \\ 
& \hspace{-0.5cm} \forall S^{r}\in\mathbb{\sigma}:  \sum_{P\in\mathcal{P}(r)}y_{P}^{r} \leq 1\label{eq:dual_fractional_LP:c1}\\ 
& \hspace{-0.5cm} \forall l\in\mathcal{L}:  \sum_{S^{r}\in\sigma}\sum_{P\in\mathcal{P}(r)|l \in P}d^{r}y_{P}^{r}\leq B(l)
\label{eq:dual_fractional_LP:c2}\\ 
& \hspace{-0.5cm} \forall n\in\mathcal{N}:  \sum_{S^{r}\in\sigma}\sum_{P\in\mathcal{P}(r)|n \in P}C(f^{r})y_{P}^{r}\leq C(n)
\label{eq:dual_fractional_LP:c3}\\ 
& \hspace{-0.5cm} \forall S^{r}\in\mathbb{\sigma},  P\in\mathcal{P}(r):\,y_{P}^{r}\geq0.
\label{eq:dual_fractional_LP:c4} 
\end{flalign}

\end{subequations}

\end{mdframed}If all service requests in $\sigma$ are known a priori,
solving the formulation in \eqref{eq:dual_fractional_LP} yields the
optimal splittable all-or-something/all-or-something packing configuration
for all the accepted services from $\sigma$. In \eqref{eq:dual_fractional_LP:obj},
we maximize the overall profit function accrued by the accepted service
requests. 

The first set of constraints in \eqref{eq:dual_fractional_LP:c1}
requires that the sum of the fractional allocations for each service
request along all possible paths is bounded above by unity. Constraints
\eqref{eq:dual_fractional_LP:c2} and \eqref{eq:dual_fractional_LP:c3}
represent the transmission and processing resource constraints on
the physical links and the NFV nodes, respectively. 

In the context of the online version of the problem, service requests
in $\sigma$ are revealed over time (in discrete steps). The idea
is to develop an online solution that maintains a feasible set whenever
a new service request arrives \textit{in a controlled manner} to
guarantee certain competitive performance. This is achieved by first
deriving the primal of \eqref{eq:dual_fractional_LP}. Second, we
need to ensure that the online algorithm produces solutions such that
the objective function of the primal and dual are bounded, which will
be explained in Subsection \ref{sec:the_approach:exponential_weights}.

Next, we present the corresponding primal formulation in \eqref{eq:primal_fractional_LP}.
Given \eqref{eq:dual_fractional_LP:c1}, we assign variable $z^{r}$
for each request $S^{r}$, where $z^{r}\in[0,\max\{\varrho^{r},\rho^{r}\}]$.
Given \eqref{eq:dual_fractional_LP:c2} and \eqref{eq:dual_fractional_LP:c3},
we assign variables $\bar{x}(l)$ and $\tilde{x}(n)$ for each physical
link $l$ ($\in\mathcal{L}$) and NFV node $n$ ($\in\mathcal{N}$),
respectively, where $\bar{x}(l)\in[0,|\mathcal{D}|_{\max}]$, $\tilde{x}(n)\in[0,\frac{\eta_{\max,}}{\eta_{\min}}]$,
$\eta_{\max}=\max_{S^{r}\in\sigma}\eta^{r}$, $\eta_{\min}=\min_{S^{r}\in\sigma}\eta^{r}$,
and $|\mathcal{D}|_{\max}=\max_{S^{r}\in\sigma}|\mathcal{D}^{r}|$.
Through the tableau method, the primal formulation is expressed in
\eqref{eq:primal_fractional_LP}.

\begin{mdframed}

\center{Primal}

\begin{subequations}
\label{eq:primal_fractional_LP}

\begin{align}
& \hspace{-0.5cm} \min  \sum_{l\in\mathcal{L}}B(l)\,\bar{x}(l)+\sum_{n\in\mathcal{N}}C(n)\,\tilde{x}(n)+\sum_{S^{r}\in\sigma}z^{r}
\label{eq:primal_fractional_LP:obj}\\
& \hspace{-0.5cm} \mathrm{subject\,to}:\nonumber\\
& \hspace{-0.5cm} \forall S^{r}\in\sigma, P\in\mathcal{P}(r):\sum_{l\in P\cap\mathcal{L}}d^{r}\bar{x}(l)+z^{r}\geq \alpha \varrho^{r}\label{eq:primal_fractional_LP:c1}\\ 
& \hspace{-0.5cm} \forall S^{r}\in\sigma,  P\in\mathcal{P}(r):\sum_{n\in P\cap\mathcal{N}}\tilde{x}(n)+z^{r}\geq \beta \rho^{r}
\label{eq:primal_fractional_LP:c2}\\ 
& \hspace{-0.5cm} \forall S^{r}\in\sigma, l\in\mathcal{L}, n\in\mathcal{N}: z^{r},\bar{x}(l),\tilde{x}(n)\geq0.
\label{eq:primal_fractional_LP:c3} 
\end{align}

\end{subequations}
\end{mdframed} In what follows, we develop an admission mechanism that is based
on the primal-dual formulation in \eqref{eq:dual_fractional_LP} and
\eqref{eq:primal_fractional_LP}.

\section{Primal-Dual based Admission Mechanism}

\label{sec:the_approach}

\subsection{The Approach}

\label{sec:the_approach:exponential_weights}

In this subsection, we provide a systematic approach to deriving
the operational cost model of the physical links and NFV nodes as
well as the admission mechanism. Let $J$ and $D$ be the value of
the objective function of the primal and the dual solutions as a result
of the online algorithm, respectively. Using the weak duality concept,
for the aforementioned primal-dual formulation, we know that $D\leq J$.
Therefore, in order to have a provable competitive ratio, we need
to bound the objective functions of the primal and dual such that
\begin{equation}
J\leq2\xi D,\label{eq:primal_dual_duality_org}
\end{equation}
while maintaining the constraints of the primal and dual formulations
satisfied, where $2\xi$ will be the competitive ratio. However, service
requests arrive in an online fashion. Therefore, instead, it is sufficient
for the online algorithm to bound the change between the primal and
dual objectives whenever a new service request arrives such that \cite{TCS-024}
\begin{equation}
\frac{\partial J}{\partial y_{P}^{r}}\leq2\xi\frac{\partial D}{\partial y_{P}^{r}},\quad S^{r}\in\sigma.\label{eq:primal_dual_duality}
\end{equation}
Due to the multi-resource form of the objective functions in \eqref{eq:dual_fractional_LP:obj}
and \eqref{eq:primal_fractional_LP:obj}, we can re-express inequality
\eqref{eq:primal_dual_duality} as

\begin{equation}
\sum_{l\in\mathcal{L}}B(l)\frac{\partial\bar{x}(l)}{\partial y_{P}^{r}}+\sum_{n\in\mathcal{N}}C(n)\frac{\partial\tilde{x}(n)}{\partial y_{P}^{r}}+\frac{\partial z^{r}}{\partial y_{P}^{r}}\leq2\varphi\alpha\varrho^{r}+2\phi\beta\rho^{r},\,S^{r}\in\sigma\label{eq:primal_dual_boundexpanded}
\end{equation}
where $\xi=\max\{\varphi,\phi\},$ and $\varphi$ and $\phi$ are
some other constants. The right-hand side of inequality \eqref{eq:primal_dual_boundexpanded}
has two profit functions, each of which corresponds to a resource
type. Therefore, to satisfy inequality \eqref{eq:primal_dual_boundexpanded},
it is sufficient to find some functions, $\bar{x}(l)$, $\tilde{x}(n)$
and $z^{r}$, such that \begin{subequations}\label{eq:pd_boundedexpanded3}
\begin{align}
\sum_{l\in\mathcal{L}}B(l)\frac{\partial\bar{x}(l)}{\partial y_{P}^{r}} & \leq2\varphi\alpha\varrho^{r},\quad S^{r}\in\sigma\label{eq:pd_boundedexpanded3_1}\\
\sum_{n\in\mathcal{N}}C(n)\frac{\partial\tilde{x}(n)}{\partial y_{P}^{r}} & \leq2\phi\beta\rho^{r},\quad S^{r}\in\sigma\label{eq:pd_boundedexpanded3_2}\\
\frac{\partial z^{r}}{\partial y_{P}^{r}} & \leq0,\quad S^{r}\in\sigma\label{eq:pd_boundedexpanded3_3}
\end{align}
\end{subequations} while maintaining the constraints of the primal
and dual formulations satisfied.  

Starting with the transmission resource type, we need to satisfy \eqref{eq:pd_boundedexpanded3_1}
while maintaining feasibility in constraints \eqref{eq:dual_fractional_LP:c1}
and \eqref{eq:primal_fractional_LP:c1}. To do so, let the solution
of the first partial derivative in \eqref{eq:pd_boundedexpanded3_1}
follow the following form,

\begin{equation}
\sum_{l\in\mathcal{L}}B(l)\frac{\partial\bar{x}(l)}{\partial y_{P}^{r}}=\varphi\sum_{l\in\mathcal{L}}\big(d^{r}\bar{x}(l)+\frac{\alpha\varrho^{r}}{L}\big),\quad S^{r}\in\sigma\label{eq:approximate_slackness}
\end{equation}
where $L$ is the maximum number of hops in a path for a unicast service
(or maximum number of physical links in a tree for a multicast service),
i.e., $\sum_{l\in\mathcal{L}}1\leq L$. After substituting
\eqref{eq:approximate_slackness} in \eqref{eq:pd_boundedexpanded3_1},
we impose a requirement that $\sum_{l\in\mathcal{L}}d^{r}\bar{x}(l)\leq\alpha\varrho^{r}$
for \eqref{eq:pd_boundedexpanded3_1} to hold. Having satisfied \eqref{eq:pd_boundedexpanded3_1},
we need to derive cost function $\bar{x}(l)$ by solving the differential
equation in \eqref{eq:approximate_slackness}. Rearranging the terms
in \eqref{eq:approximate_slackness}, we have a differential equation
of the following form,
\begin{equation}
\sum_{l\in\mathcal{L}}\frac{\partial\bar{x}(l)}{\partial y_{P}^{r}}+\sum_{l\in\mathcal{L}}\frac{-\varphi d^{r}}{B(l)}\bar{x}(l)=\sum_{l\in\mathcal{L}}\frac{\varphi\alpha\varrho^{r}}{B(l)L}.\label{eq:costfun_diff2-1}
\end{equation}
Define the integrating factor
\begin{align}
I & =\exp(\sum_{S^{r}\in\sigma|l\in P\in\mathcal{P}(r)}\int\frac{-\varphi d^{r}}{B(l)}\partial y_{P}^{r})\nonumber \\
 & =C\exp(\frac{-\varphi}{B(l)}\sum_{S^{r}\in\sigma|l\in P\in\mathcal{P}(r)}d^{r}y_{P}^{r})\label{eq:integrating_factor}
\end{align}
where $C$ is an arbitrary constant, and multiply both sides of \eqref{eq:costfun_diff2-1}
by $I$, we get
\begin{equation}
I\big(\frac{\partial\bar{x}(l)}{\partial y_{P}^{r}}+\frac{-\varphi d^{r}}{B(l)}\bar{x}(l)\big)=\frac{\varphi\alpha\varrho^{r}}{B(l)L}I.
\end{equation}
Through the following identity by the product rule,
\begin{equation}
\frac{\partial}{\partial y_{P}^{r}}(Ix)=I\frac{\partial x}{\partial y_{P}^{r}}+\frac{-\varphi d^{r}}{B(l)}Ix,
\end{equation}
we can express $\bar{x}(l)$ as
\begin{align}
\nonumber \bar{x}(l) & =I^{-1}\frac{\varphi d^{r}}{B(l)L}\int I\partial y_{P}^{r},\\
& =\frac{1}{L}(-Ce^{\frac{\varphi}{B(l)}\sum_{S^{r}\in\sigma|l\in P\in\mathcal{P}(r)}d^{r}y_{P}^{r}}-1),\quad l\in\mathcal{L}.
\end{align}
Initially, before the arrival of any service request, we require that
$\bar{x}(l)=0$, which occurs when $\frac{1}{B(l)}\sum_{S^{r}\in\sigma|l\in P\in\mathcal{P}(r)}d^{r}y_{P}^{r}=0$.
Thus, we set $C=-1$. We also require that $\bar{x}(l)\geq\alpha|\mathcal{D}^{r}|^{k}$
when physical link $l$ ($\in\mathcal{L}$) is saturated, i.e. when
$\frac{1}{B(l)}\sum_{S^{r}\in\sigma|l\in P\in\mathcal{P}(r)}d^{r}y_{P}^{r}=1$.
Hence, we need $\varphi\geq\ln(\alpha L|\mathcal{D}|_{\max}^{k}+1)$.
Therefore, $\bar{x}(l)$ can be expressed as
\begin{align}
\bar{x}(l) & =\frac{1}{L}(e^{\varphi\frac{\sum_{S^{r}\in\sigma|l\in P\in\mathcal{P}(r)}d^{r}y_{P}^{r}}{B(l)}}-1),\quad l\in\mathcal{L}.
\end{align}
Note that constraint \eqref{eq:dual_fractional_LP:c2} is maintained
feasible over the whole range of $\bar{x}(l)$. In the admission mechanism,
we need to ensure that there is a sufficient protection to the resources
before accepting a service to avoid the scenario of accidentally violating
the resources (which occurs when $\bar{x}(l)\geq\alpha|\mathcal{D}^{r}|$).
It turns out that, due to Lemma \ref{lemma:trans} (to be derived
later), if the required data rate is bounded above by $d^{r}\leq\frac{\min_{l\in\mathcal{L}}B(l)}{\varphi}$,
we need $\varphi\geq\ln(2\alpha L|\mathcal{D}|_{\max}^{k}+2)$. Note
that the edge costs include variables from future requests (i.e.,
$y_{P}^{r},\,\forall S^{r}\in\sigma$). However, since this is an
online algorithm, future variables can be initialized to zero until
the respective service requests are parsed through the admission mechanism.
We can express the edge costs in a multiplicative recursive manner
as
\begin{equation}
\bar{x}^{r}(l)=\bar{x}^{r-1}(l)e^{\varphi\frac{d^{r}}{B(l)}}+\frac{1}{L}(e^{\varphi\frac{d^{r}}{B(l)}}-1),\quad l\in\mathcal{L},\,S^{r}\in\sigma\label{eq:haha}
\end{equation}
where $\bar{x}^{r}(\cdot)$ is the edge cost after embedding the $r$th
service request, and $\bar{x}^{0}(\cdot)$ is set to zero. Now, we
need to ensure that constraint \eqref{eq:primal_fractional_LP:c1}
is maintained feasible. Since we set $\sum_{l\in\mathcal{L}}d^{r}\bar{x}(l)\leq\alpha\varrho^{r}$,
we require that $z^{r}\geq\alpha\varrho^{r}-\sum_{l\in\mathcal{L}}d^{r}\bar{x}(l)$
for \eqref{eq:primal_fractional_LP:c1} to be maintained feasible,
and for \eqref{eq:pd_boundedexpanded3_3} to hold.

By following a similar procedure, for the processing resource type,
we need to satisfy \eqref{eq:pd_boundedexpanded3_2} while maintaining
feasibility in constraints \eqref{eq:dual_fractional_LP:c3} and \eqref{eq:primal_fractional_LP:c2}.
To do so, let the solution of the partial derivative in \eqref{eq:pd_boundedexpanded3_2}
follow the following form,
\begin{equation}
\sum_{n\in\mathcal{N}}C(n)\frac{\partial\tilde{x}(n)}{\partial y_{P}^{r}}=\phi\sum_{n\in\mathcal{N}}\big(C(f^{r})\tilde{x}(n)+\frac{\beta\rho^{r}}{K}\big),\quad S^{r}\in\sigma\label{eq:solve}
\end{equation}
where $K$ is the maximum number of NF instances in a service request
(including both best-effort and mandatory NF instances), i.e., $\sum_{n\in\mathcal{N}}1\leq K$.
Similarly, we can satisfy \eqref{eq:pd_boundedexpanded3_2} by imposing
a condition that $\sum_{n\in\mathcal{N}}C(f^{r})\tilde{x}(n)\leq\beta\rho^{r}$.
Now, we need to solve differential equation \eqref{eq:solve}. A similar
procedure yields a cost function that is expressed as
\begin{equation}
\tilde{x}(n)=\frac{1}{K}(-Ce^{\phi\frac{\sum_{S^{r}\in\sigma|n\in P\in\mathcal{P}(r)}C(f^{r})y_{P}^{r}}{C(n)}}-1),\quad n\in\mathcal{N}.
\end{equation}
Initially, before the arrival of any service request, we require that
$\tilde{x}(n)=0$, which occurs when $\frac{1}{C(n)}\sum_{S^{r}\in\sigma}C(f^{r})y_{P}^{r}=0$.
Therefore, we have $C=-1$. We also require that $\tilde{x}(n)\geq\beta\eta^{r}$
when NFV node $n$ ($\in\mathcal{N}$) is saturated, i.e. when $\frac{1}{C(n)}\sum_{S^{r}\in\sigma}C(f^{r})y_{P}^{r}=1$.
Therefore, we need $\phi\geq\ln(\beta K\frac{\eta_{\max}}{\eta_{\min}}+1)$.
Hence, $\tilde{x}(n)$ can be expressed as
\begin{align}
\tilde{x}(n) & =\frac{1}{K}(e^{\phi\frac{\sum_{S^{r}\in\sigma|n\in P\in\mathcal{P}(r)}C(f^{r})y_{P}^{r}}{C(n)}}-1),\quad n\in\mathcal{N}.
\end{align}
Note that constraint \eqref{eq:dual_fractional_LP:c3} is maintained
feasible over the whole range of $\tilde{x}(n)$. Similar to the transmission
resource, due to Lemma \ref{lemma_proc} (to be given in Subsection
\ref{sec:perf_analysis}), if the required processing rate is bounded
above by $C(f^{r})\leq\frac{\min_{n\in\mathcal{N}}C(n)}{\phi}$, we
need $\phi\geq\ln(2\beta K\frac{\eta_{\max}}{\eta_{\min}}+2)$ to
ensure a sufficient protection against violating the processing resources.
Since we set $\sum_{n\in\mathcal{N}}C(f^{r})\tilde{x}(n)\leq\beta\rho^{r}$,
we require that $z^{r}\geq\max\{\alpha\varrho^{r}-\sum_{l\in\mathcal{L}}d^{r}\bar{x}(l),\,\beta\rho^{r}-\sum_{n\in\mathcal{N}}C(f^{r})\bar{x}(n)\}$
for \eqref{eq:primal_fractional_LP:c1} and \eqref{eq:primal_fractional_LP:c2}
to be maintained feasible, and for \eqref{eq:pd_boundedexpanded3_3}
to hold. The cost function can be updated multiplicatively each time
after embedding a service request as follows, 
\begin{equation}
\tilde{x}^{r}(n)=\tilde{x}^{r-1}(n)e^{\phi\frac{C(f^{r})}{C(n)}}+\frac{1}{K}(e^{\phi\frac{C(f^{r})}{C(n)}}-1),\quad n\in\mathcal{N},\,S^{r}\in\sigma
\end{equation}
where $\tilde{x}^{r}(n)$ is the cost of the NFV node $n$ ($\in\mathcal{N}$)
after embedding the $r$th service request, and $\tilde{x}^{0}(n)$
is set to zero. Now, we are ready to state the all-or-nothing/all-or-something
admission mechanism.

\subsection{Admission Mechanism}

\label{sec:PD_adm_policy}

The procedure commences with the arrival of an $r$th service request,
which resembles an augmentation of a new decision variable ($y_{P}^{r}$)
to the dual formulation. Correspondingly, in the primal formulation,
the arrival is equivalent to the augmentation of two new constraints,
namely \eqref{eq:primal_fractional_LP:c1} and \eqref{eq:primal_fractional_LP:c2}.
 The $r$th service request is accepted if there exists a path, $P$\footnote{How to find path $P$ for the routing and NF placement problem is
addressed in Subsection \ref{sec:routing_solution}.}, such that the following two conditions hold:
\begin{equation}
\sum_{l\in P\cap\mathcal{L}}d^{r}\bar{x}^{r-1}(l)\leq\alpha\varrho^{r}\label{eq:admission_transmission}
\end{equation}
and
\begin{equation}
\sum_{n\in P\cap\mathcal{N}}C(f^{r})\tilde{x}^{r-1}(n)\leq\beta\rho^{r}.\label{eq:admission_processing}
\end{equation}
If the two conditions are satisfied, accept the request and route
it on $P$, and set $y_{P}^{r}=1$. To maintain feasibility in \eqref{eq:primal_fractional_LP:c1}
and \eqref{eq:primal_fractional_LP:c2}, set
\begin{equation}
z^{r}=\max\big(\alpha\varrho^{r}-\sum_{l\in P\cap\mathcal{L}}d^{r}\bar{x}(l),\beta\rho^{r}-\sum_{l\in P\cap\mathcal{N}}C(f^{r})\tilde{x}(n)\big).\label{eq:update_z}
\end{equation}
Finally, update the costs of the edge variables $\bar{x}(l)$ and
NFV nodes $\tilde{x}(n)$ in a multiplicative manner as follows,
\begin{equation}
\bar{x}^{r}(l)=\bar{x}^{r-1}(l)e^{\varphi\frac{d^{r}}{B(l)}}+\frac{1}{L}(e^{\varphi\frac{d^{r}}{B(l)}}-1),\quad l\in P\cap\mathcal{L}\label{eq:cost_trans}
\end{equation}

\begin{equation}
\tilde{x}^{r}(n)=\tilde{x}^{r-1}(n)e^{\phi\frac{C(f^{r})}{C(n)}}+\frac{1}{K}(e^{\phi\frac{C(f^{r})}{C(n)}}-1),\quad n\in P\cap\mathcal{N}\label{eq:cost_proc}
\end{equation}
where $\varphi=\ln(2\alpha L|\mathcal{D}|_{max}^{k}+2)$ and $\phi=\ln(2\beta K\frac{\eta_{\max}}{\eta_{\min}}+2)$.

\textit{Treatment of best-effort NFs} -- We first
find a routing and NF placement configuration that includes the set
of best-effort NFs since it provides a larger (incentivized) profit
with $\rho^{r}=C(f^{r})\eta_{b}^{r}$. If such configuration is rejected
by the admission mechanism, we find another configuration that excludes
the set of best-effort NFs, and check against the admission mechanism
with the nominal profit function $\rho^{r}=C(f^{r})\eta_{m}^{r}$.
If both configurations (with and without the set of best-effort NFs)
do not satisfy \eqref{eq:admission_transmission} and \eqref{eq:admission_processing},
the service request is rejected. Fig. \ref{fig:big_image} provides
a representation of the all-or-nothing/all-or-something admission
mechanism. 
\begin{figure}
\begin{centering}
\includegraphics[width=0.4\textwidth]{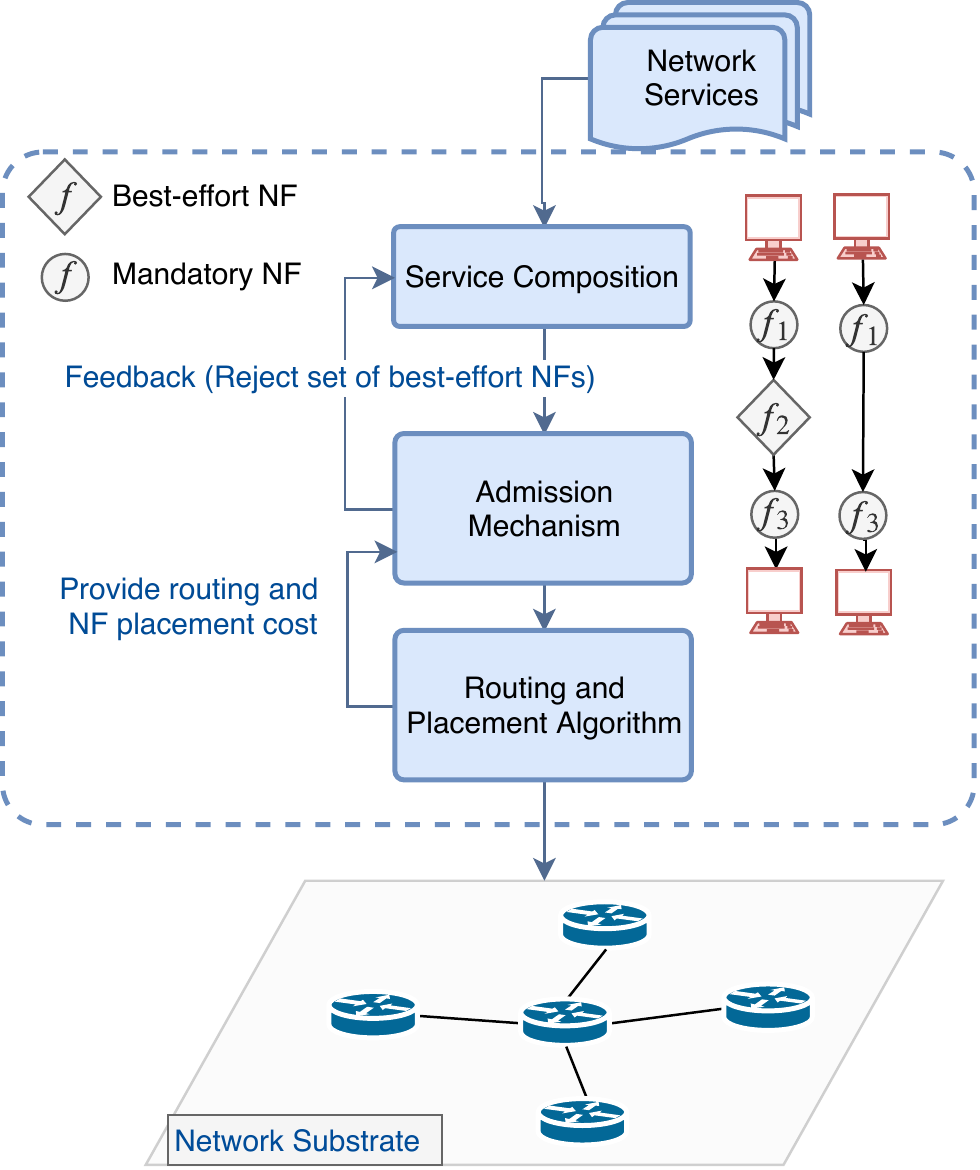}
\par\end{centering}
\caption{The joint all-or-nothing/all-or-something admission mechanism and
online JCRP framework.}
\label{fig:big_image}
\end{figure}
Algorithm \ref{alg:JPR1} summarizes the online admission framework.
In the algorithm's pseudocode, an assignment is denoted by ``$\leftarrow$''.
\begin{algorithm}[htbp]
\SetKwInOut{Input}{Input}
\SetKwInOut{Output}{Output}
\underline{Procedure} OnlineJPR$(\mathcal{G},{S}^{r})$\;

\Input{$\mathcal{G}(\mathcal{N},\mathcal{L})$, $S^{r}$}
\Output{Embed or reject $S^r$}

$\bar{x}(l) \leftarrow 0$ (\textit{only once});
$\tilde{x}(n) \leftarrow 0$ (\textit{only once})\;
services $\leftarrow$ ($S^r$, $S^r - \mathcal{V}_{b}^{r}$)\;
\For{$i = 1:2$}{

$P \leftarrow$ JRP($\mathcal{G}$, services[$i$]);
\Comment JRP is in Algorithm \ref{alg:JPR_unicast} in Section \ref{sec:routing_solution}.

\If{$\sum_{l\in P\cap\mathcal{L}} d^{r} \bar{x}(l)  \leq \alpha \varrho^{r} $ and $\sum_{n\in P\cap\mathcal{N}} C(f^{r}) \tilde{x}(n)  \leq \beta \rho^{r} $}{
	Accept request ($y_P^r \leftarrow 1$)\;
	$z^{r}\leftarrow\max\big(\alpha \varrho^{r}-\sum_{l\in P\cap\mathcal{L}}d^{r}\bar{x}(l),\beta \rho^{r}-\sum_{n\in P\cap\mathcal{N}}C(f^{r})\tilde{x}(n)\big)$\;

	$\bar{x}(l)\leftarrow\bar{x}(l)e^{\varphi \frac{d^{r}}{B(l)}}+\frac{1}{L}(e^{\varphi \frac{d^{r}}{B(l)}}-1),\quad \forall l\in P\cap\mathcal{L}$\; 

	$\tilde{x}(n)\leftarrow\tilde{x}(n)e^{\phi \frac{C(f^{r})}{C(n)}}+\frac{1}{K}(e^{\phi \frac{C(f^{r})}{C(n)}}-1),\quad \forall n\in P\cap\mathcal{N}$\;
	\Return\;
}
}
\Return (Reject request)\;

\caption{Admission control and online JCRP framework}
\label{alg:JPR1}
\end{algorithm}

\subsection{Performance Analysis}

\label{sec:perf_analysis} In what follows, we analyze the performance
of the admission mechanism. We show that the proposed mechanism does
not violate the transmission and processing resources of physical
links and NFV nodes. Then, we prove the competitive ratio for the
all-or-nothing/all-or-something admission mechanism.
\begin{lem}
The transmission resource constraint on the physical links cannot
be violated, i.e., $\ensuremath{\sum_{S^{r}\in\sigma|l\in P\in\mathcal{P}(r)}d^{r}\leq B(l),\;l\in\mathcal{L}}$,
if $\varphi\geq\ln(2\alpha L|\mathcal{D}|_{\max}^{k}+2)$ and $d^{r}\leq\frac{\min_{l\in\mathcal{L}}B(l)}{\varphi}$,
$\forall S^{r}\in\sigma$.\label{lemma:trans} 
\end{lem}
\begin{lem}
The processing resource constraint on the NFV nodes cannot be violated,
i.e. $\allowbreak\ensuremath{\sum_{S^{r}\in\sigma|n\in P\in\mathcal{P}(r)}C(f^{r})\leq C(n),\;n\in\mathcal{N}}$,
if $\phi\geq\ln(2\beta K\frac{\eta_{\max}}{\eta_{\min}}+2)$ and $C(f^{r})\leq\frac{\min_{n\in\mathcal{N}}C(n)}{\phi}$,
$\forall S^{r}\in\sigma$.\label{lemma_proc}
\end{lem}
\begin{thm}
The competitive ratio of the admission mechanism is $\mathcal{O}\big(\max(\varphi,\phi)\big)$,
where $d^{r}\leq\frac{\min_{l\in\mathcal{L}}B(l)}{\varphi}$, $C(f^{r})\leq\frac{\min_{n\in\mathcal{N}}C(n)}{\phi}$,
$\varphi=\ln(2\alpha L|\mathcal{D}|_{\max}^{k}+2)$, and $\phi=\ln(2\beta K\frac{\eta_{\max}}{\eta_{\min}}+2)$.\label{proof_compet}
\end{thm}
The proof of Lemmata \ref{lemma:trans} and \ref{lemma_proc} and
Theorem \ref{proof_compet} is given in Appendix.

Now, we discuss the method to find path $P$ for a service request.
Recall the admission conditions in \eqref{eq:admission_transmission}
and \eqref{eq:admission_processing}. For each service request, the
admission mechanism requires checking \textit{all} possible paths
for the performance guarantees to hold. If any of such paths satisfies
the admission mechanism, the service request should be accepted. Notably,
in general, it is not necessary to route the service on the minimum-cost
path, but rather a secondary routing objective can be invoked. However,
there exists an exponential number of possible paths for a service
request. It is more convenient (and sufficient) to check against the
minimum-cost path only. If it was rejected, all other paths would
be rejected. Next, we propose an algorithm to find the minimum-cost
routing and NF placement solution for a unicast and multicast service
request.

\section{Routing and NF Placement Approximation Algorithm}

\label{sec:routing_solution}

Due to Lemmata \ref{lemma:trans} and \ref{lemma_proc}, the proposed
admission mechanism guarantees that a violation in the processing
and transmission resources is always avoided if the exponential cost
functions are used for the physical links and NFV nodes with $\varphi\geq\ln(2\alpha L|\mathcal{D}|_{\max}^{k}+2)$
and $\phi\geq\ln(2\beta K\frac{\eta_{\max}}{\eta_{\min}}+2)$. Therefore,
it is sufficient to design a routing and NF placement algorithm for
the unconstrained (or uncapacitated) scenario. Here, we propose a
\textit{one-step} algorithm for the routing and NF placement of unicast
and multicast services for the unconstrained scenario. The algorithm
relies mainly on the construction of an auxiliary multilayer network
transformation that has a one-to-one mapping from the NF placement
and routing problem to an equivalent routing problem. This facilitates
the use of existing (approximation) algorithms, such as the Dijkstra
shortest path for the unicast scenario and MST-based Steiner tree
for the multicast scenario.

\subsection{Auxiliary Network Transformation and Routing and NF Placement Algorithm}

To jointly consider the provisioning costs of both NF (processing)
and virtual links (transmission), for each service request, we construct
an auxiliary multilayer graph from the network substrate, in which
the constructed edges represent either (i) the hosting of a virtual
link or (ii) the processing of some NF type. For the sake of exposition,
since the joint routing and NF placement algorithm treats each service
request separately, we drop superscript $r$ (which alludes to the
$r$th service request) in this subsection.

The auxiliary multilayer graph is modeled as a directed graph $\mathcal{G}_{M}=(\mathcal{N}_{M},\mathcal{L}_{M})$,
where $\mathcal{N}_{M}\subseteq\mathcal{N}\times\mathcal{X}$ is the
set that contains all nodes, in which node $n$ ($\in\mathcal{N}$)
is present in a corresponding layer $a$ ($\in\mathcal{X}$); denote
such a node by $n^{\alpha}$. Correspondingly, $\mathcal{L}_{M}\subseteq\mathcal{N}_{M}\times\mathcal{N}_{M}$
is the set of all inter- and intra-layer edges. Intra-layer edges,
$\mathcal{L}_{A}=\{(u^{a},v^{b})\in\mathcal{L}_{M}|a=b\in\mathcal{X}\}$,
represent the routing connections between the network elements (nodes)
in each layer. Inter-layer edges, $\mathcal{L}_{I}=\{(u^{a},u^{b})\in\mathcal{L}_{M}|a\neq b\in\mathcal{X}\}$,
are used to encode the placement decisions in which a traversal of
an edge from one layer $a$ ($\in\mathcal{X}$) to another layer $b$
($\in\mathcal{X}$) maps to the processing of the $a$th NF instance
in a service request. Hence, for service request $S^{r}$, the number
of layers is equivalent to the number of NFs plus one (i.e., $|\mathcal{X}|=|\mathcal{V}|+1$).

Upon the arrival of the $r$th service request, we construct a multilayer
graph $\mathcal{G}_{M}$ as follows:
\begin{enumerate}
\item Create $|\mathcal{X}|$ ($=|\mathcal{V}|+1$) copies of the network
substrate $\mathcal{G}$. Each copy ($\mathcal{G}^{i}(\mathcal{N}^{i},\mathcal{L}^{i})$)
represents one layer, where $i=\{0,\dots,|\mathcal{V}|\}$. The transmission
and processing resources for each NFV node ($n\in\mathcal{N}$) and
physical link $(l\in\mathcal{L}$) are equal in each copy; 
\item Assign the source node to its corresponding node at layer 0 ($=s^{0}$,
$s^{0}\in\mathcal{N}^{0}$);
\item Assign the destination nodes to their corresponding nodes at the last
layer ($=t^{|\mathcal{V}|}$, $t^{|\mathcal{V}|}\in\mathcal{N}^{|\mathcal{V}|}$);
\item For the first $|\mathcal{V}|$ layers, construct inter-layer edges
from layer $i$ to layer $i+1$ ($l\leftarrow(n^{i},n^{i+1})$) for
each NFV node that can host $f_{i}$ (i.e., if $n^{i}\in\mathcal{F}_{i}$).
The processing resources of each inter-layer edge, $l\leftarrow(n^{i},n^{i+1})$,
is that of the corresponding NFV node $n^{i}$. 
\end{enumerate}
We provide an illustrative example of the construction of the auxiliary
graph transformation in Figs. \ref{fig:problem_input} and \ref{fig:problem_output}.
Fig. \ref{fig:problem_input} illustrates a unicast service request
of two NFs ($f_{1}$ and $f_{2}$), where the source is $n_{1}$ and
the destination is $n_{4}$. We also have a network substrate of 4
NFV nodes that can host either NF type or both. Fig. \ref{fig:problem_output}
illustrates the construction of the auxiliary graph transformation.
Since we have two NFs, the transformation has three layers. NFV nodes
$n_{1}$ and $n_{3}$ can host $f_{1}$. Therefore, we construct the
inter-layer edges, $n_{1}^{1}\rightarrow n_{1}^{2}$ and $n_{3}^{1}\rightarrow n_{3}^{2}$.
Similarly, $n_{1}$ and $n_{2}$ can host $f_{2}$. Therefore, we
construct the inter-layer edges, $n_{1}^{2}\rightarrow n_{1}^{3}$
and $n_{2}^{2}\rightarrow n_{2}^{3}$.

\begin{figure}
\begin{centering}
\includegraphics[width=0.45\textwidth]{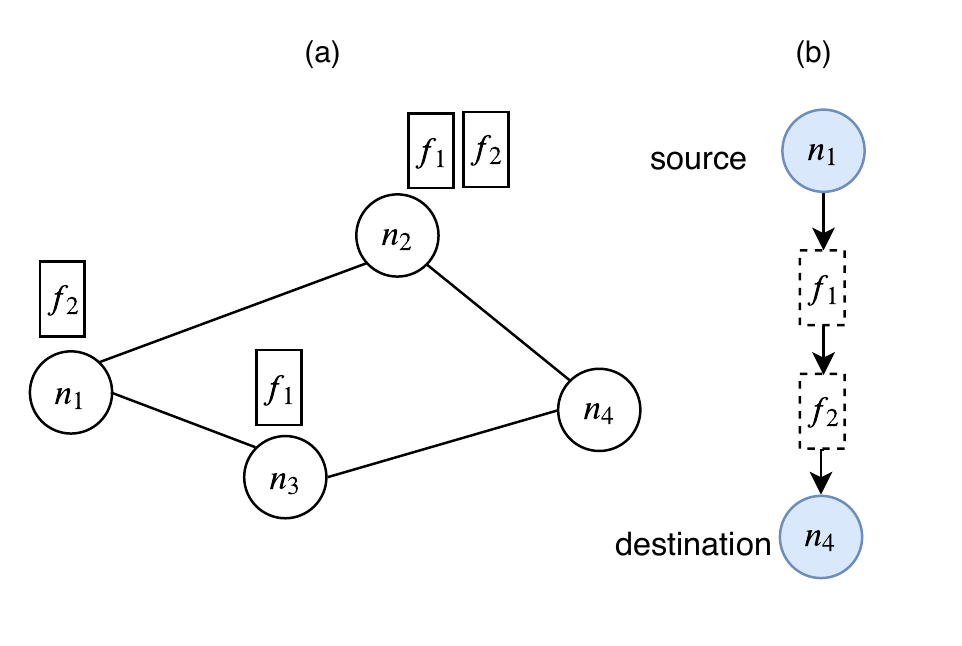}
\par\end{centering}
\caption{A problem input: (a) A network substrate along with the permissible
NFs on each network element, and (b) the logical topology of a service
request.}
\label{fig:problem_input}
\end{figure}
\begin{figure}
\begin{centering}
\includegraphics[width=0.45\textwidth]{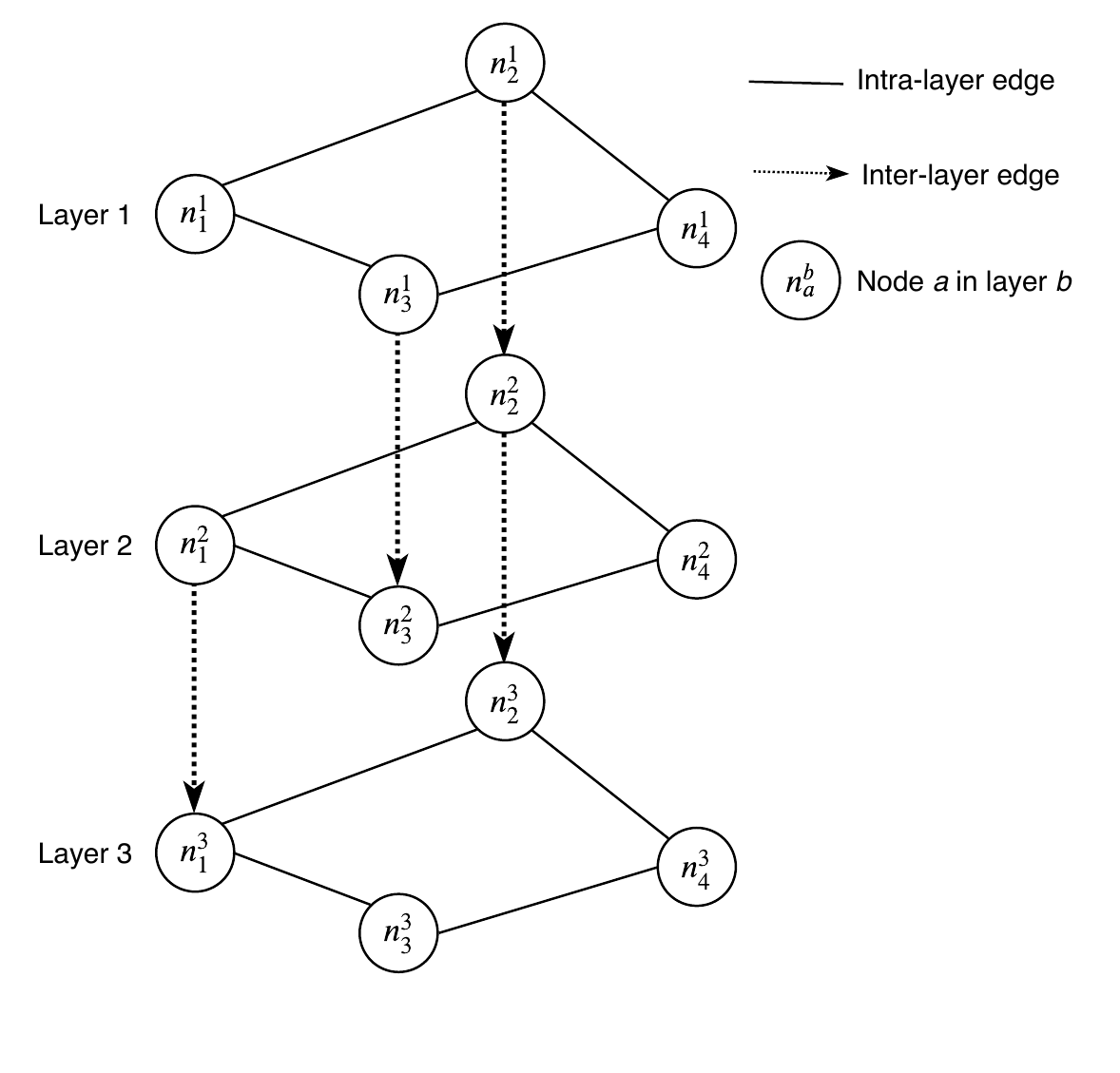}
\par\end{centering}
\caption{The auxiliary network transformation for the problem input in Fig.
\ref{fig:problem_input}.}
\label{fig:problem_output}
\end{figure}

With the new auxiliary graph transformation, a path traversal (while
considering only the edge costs) from node $s^{0}$ to node $t^{|\mathcal{V}|}$
represents a routing and NF placement solution in the original network
substrate graph. Algorithm \ref{alg:JPR_unicast} summarizes the online
JRP framework, which comprises the construction of the auxiliary graph
transformation with the minimum-cost routing algorithm for a unicast
or multicast service. For multicast services, the MST-based Steiner
tree algorithm is a 2-approximation algorithm. Therefore, the competitive-ratio
of the admission mechanism for the multicast services is $\mathcal{O}(2\max\{\varphi,\phi\})=\mathcal{O}(\max\{\varphi,\phi\})$. 

\begin{algorithm}[htbp]
\SetKwInOut{Input}{Input}
\SetKwInOut{Output}{Output}
\underline{Procedure} JRP$(\mathcal{G},{S}^{r})$\;

\Input{$\mathcal{G}$, $S^{r} = S =(s,\mathcal{D},f_{1},f_{2},\dots, f_{|\mathcal{V}|},d)$}
\Output{$P$}

$\{\mathcal{G}^{k}(N^{k},L^{k})\}_{k=0}^{|\mathcal{V}|+1}$ $\leftarrow$  $\mathcal{G}(N,L)$\;

$s^{0}$ $\leftarrow$  $s$ (source node at layer 0 as source $s$)\;
$\mathcal{D}^{|\mathcal{V}|}$ $\leftarrow$  $\mathcal{D}$\;

$\mathcal{L}_I \leftarrow \{\}$\;
\For{$k=0:(|\mathcal{V}|-1)$}{
	\For{$n^{k} \in \mathcal{F}_{k}$}{
	Add $l \leftarrow (n^{k},n^{k+1})$ to $\mathcal{L}_I$\; 
	$C_l(r) = C(n^{k})$\;
	}
}

For multicast services, find an MST-based Steiner tree from $s^{0}$ to $\mathcal{D}^{|\mathcal{V}|}$, while utilizing the cost functions in \eqref{eq:cost_trans} and \eqref{eq:cost_proc}, and save on $P$\;
For unicast services, find a Dijkstra shortest path from $s^{0}$ to $t^{|\mathcal{V}|}$, while utilizing the cost functions in \eqref{eq:cost_trans} and \eqref{eq:cost_proc}, and save on $P$\;

\Return $P$\;

\caption{JRP algorithm for a single service request}
\label{alg:JPR_unicast}
\end{algorithm}

\begin{thm}
For unicast services, the overall time complexity of the online routing
and NF placement framework is $\mathcal{O}(h\log h)$, where $h=|\mathcal{N}|(K+1).$
For multicast services, the overall time complexity is $\mathcal{O}\big(|\mathcal{D}|^{2}h\big)$.
\end{thm}
\begin{proof}
For an efficient run-time, the full construction of the auxiliary
network transformation can be performed once in the beginning. With
the arrival of a new service request, depending on the requested NFs,
the relevant layers can be connected with each other, while temporarily
deactivating irrelevant layers (by disconnecting them). Therefore,
the major component of the runtime is due to the Dijkstra shortest
path for unicast services and MST-based Steiner tree for multicast
services that is run over the network transformation which has $|\mathcal{N}|(K+1)$
nodes.
\end{proof}

\section{Discussions and Simulation Results}

\label{sec:numerical}

\subsection{Discussions}

\label{sec:numerical:discussion}

\textit{On the obtained competitive ratio} \textendash{} As seen
from the derivations in Subsection \ref{sec:the_approach:exponential_weights},
in order to bound the performance of the primal and the dual, exponential
cost functions are used for the physical links and NFV nodes. In doing
so, if we can guarantee that the residual resources for the physical
links and NFV nodes are not violated, the admission mechanism yields
a competitive ratio of $\mathcal{O}\big(\max\{\varphi,\phi\}\big)=O\big(\max\{\ln\alpha L|\mathcal{D}|_{\max}^{k},\ln\beta K\frac{\eta_{\max}}{\eta_{\min}}\}\big)$.
However, the routing and NF placement problem for the constrained
scenario is NP-hard. Therefore, to protect against a possible violation
in the resources without relying on a constrained routing and NF placement
algorithm, Lemmata \ref{lemma:trans} and \ref{lemma_proc} require
that $\varphi=\ln(2\alpha L|\mathcal{D}|_{\max}^{k}+2)$ and $\phi=\ln(2\beta K\frac{\eta_{\max}}{\eta_{\min}}+2)$
at least, for which the competitive ratio is increased to $\mathcal{O}\big(\max\{\ln2\alpha L|\mathcal{D}|_{\max}^{k},\ln2\beta K\frac{\eta_{\max}}{\eta_{\min}}\}\big)$.
A consequential drawback is that the utilization of physical links
and NFV nodes will not exceed $1-\frac{1}{\varphi}$ and $1-\frac{1}{\phi}$,
respectively. Therefore, when $\varphi$ or $\phi$ are relatively
small, a considerable amount of processing and transmission resources
will be wasted. Hence, the online algorithm is expected to not perform
well for networks of a small size relative to the intended competitive
performance.

\textit{On the design of profit function} \textendash{} The proposed
framework works for both unicast and multicast services. Moreover,
it includes both best-effort and mandatory NF types. The profit functions
allows a variation that depends on the maximum number of included
destinations and the maximum incentive for including the set of best-effort
NFs. To this effect, we can observe that the optimality is penalized
due to the large variation of the profit functions, where maximizing
the amortized throughput only (i.e., with $\varrho^{r}=d^{r}$ and
$\rho^{r}=C(f^{r}$)) improves the competitive ratio by a logarithmic
factor in $|\mathcal{D}|_{\max}^{k}$ and $\frac{\eta_{\max}}{\eta_{\min}}$,
respectively.

Recall that the derived competitive ratio is $\mathcal{O}\big(\max\{\ln2\alpha L|\mathcal{D}|_{\max},\ln2\beta K\frac{\eta_{\max}}{\eta_{\min}}\}\big)$.
In practice, $L|\mathcal{D}|_{\max}^{k}$ is larger than the maximum
number of NFs ($K$). Therefore, the use of an incentive for including
the best-effort NFs can help to scale up the second term without necessarily
degrading the competitive performance, which offers an appropriate
generalization.

\subsection{Numerical Analysis}

\label{sec:numerical:numerical}

In this subsection, we analyze three online algorithms. The first
algorithm is the proposed approximation algorithm with $\varphi=\ln(2\alpha L|\mathcal{D}|_{\max}^{k}+2)$
and $\phi=\ln(2\beta K\frac{\eta_{\max}}{\eta_{\min}}+2)$ (as in
Algorithm 1). As shown, a resource violation is always avoided, and
the competitive performance is guaranteed. The second online algorithm
is similar to the first one but with $\varphi=\ln(\alpha L|\mathcal{D}|_{\max}^{k}+1)$
and $\phi=\ln(\beta K\frac{\eta_{\max}}{\eta_{\min}}+1)$. Here, resources
are not necessarily protected from future violations. As a heuristic
algorithm, if the routing and NF placement solution violates any processing
or transmission constraint, it is removed. The third online algorithm
is a greedy algorithm that attempts to accept all services as long
as there are sufficient resources. The greedy algorithm basically
resembles the heuristic algorithm but without invoking the admission
conditions in \eqref{eq:admission_transmission} and \eqref{eq:admission_processing}.
That is, it is basically the output of Algorithm 2 (without Algorithm
1), and with an extra step of checking if the service request violates
any processing or transmission constraint. 

In the experiments, we analyze the performance of the three algorithms
on linear, random, and real network substrate topologies. Throughout
the experiments, we set the scalarization coefficients ($\alpha,\beta$)
to unity (since the processing and transmission resources are appropriately
scaled). Each NFV node can host 2/3 of the possible NFs in random.
The transmission and processing resources are randomly distributed
between 1000 and 5000 packet/s. The required data rate for service
requests is uniformly distributed between 1 and 20 packet/s. The processing
rate requirement of NF instances are linearly proportional to the
incoming data rate $C(f^{r})$ = $d^{r}$ \cite{Beck2015,Ye2016}.
In all trials, we terminate an algorithm when it no longer can accept
any request, i.e., when the network substrate instance reaches the
maximum possible utilization.

\begin{figure}
\begin{centering}
\includegraphics[width=0.45\textwidth]{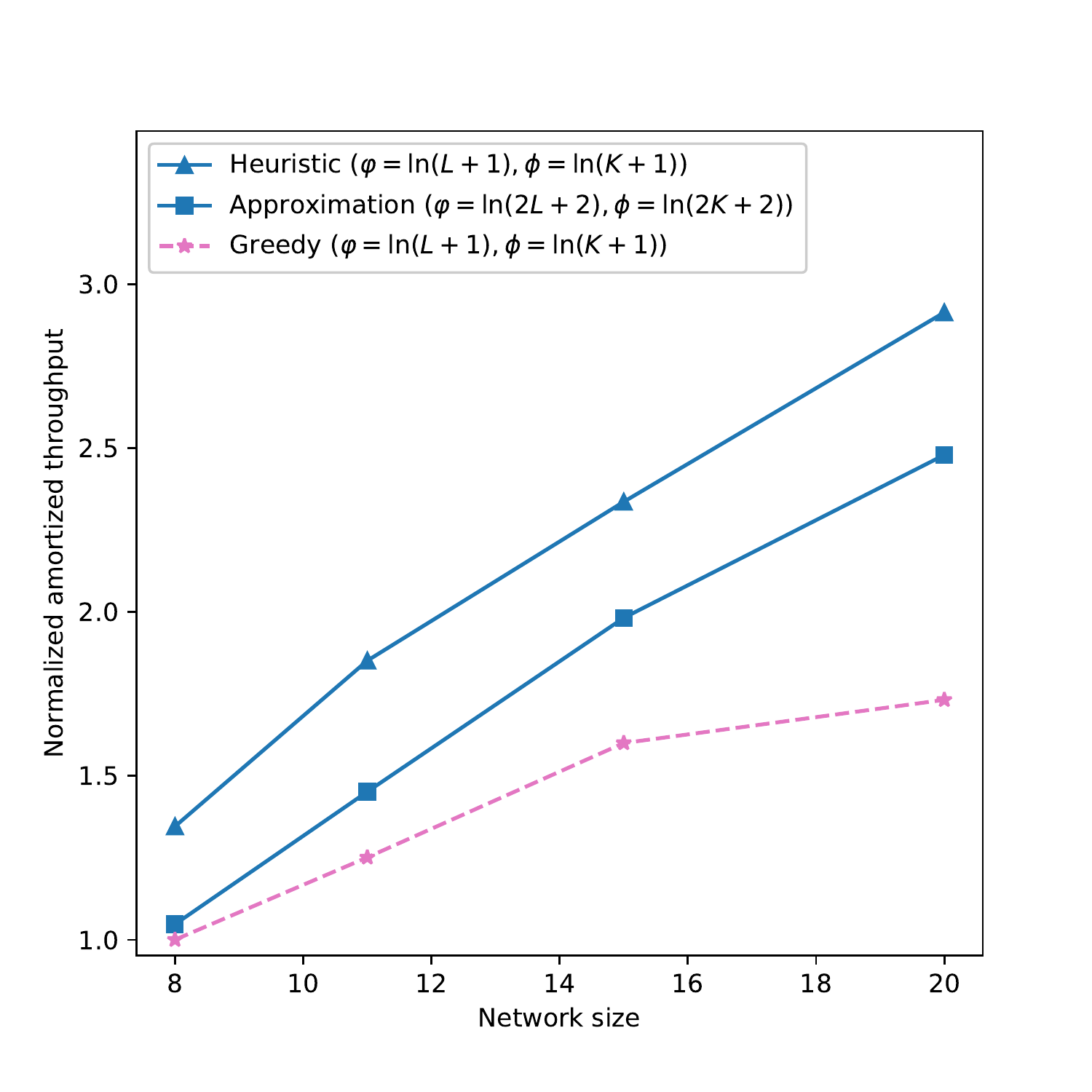}
\par\end{centering}
\caption{Normalized aggregate throughput ($\varrho^{r}+\rho^{r}$) for the
three algorithms for a linear topology, with $L=K=4$.}
\label{fig1}
\end{figure}

In the first experiment, we generate a directed network substrate
with a linear topology for unicast service requests. Each request
has a random pair of source and destination nodes, and an overall
number of required NFs ($|\mathcal{V}^{r}|$) of $3$, with the number
of best-effort NFs ($|\mathcal{V}_{b}^{r}|$) uniformly distributed
between 0 and $3$. The incentive for including the best-effort NFs
($\eta^{r}$) is set to unity. Fig. \ref{fig1} shows the normalized
aggregate profit for the three algorithms as the size of the linear
network substrate ($|\mathcal{N}|$) grows, with $L=K=4$. The aggregate
profit increases almost linearly for all the algorithms. The heuristic
algorithm (with $\varphi=\ln(L+1)$ and $\phi=\ln(K+1)$) outperforms
the approximation algorithm (with $\varphi=\ln(2L+2)$ and $\phi=\ln(2K+2)$)
by a constant gap of approximately 30\%. Interestingly, this is equal
to the wasted utilization by the latter algorithm, which is given
by $1-\frac{1}{\varphi}$ and $1-\frac{1}{\phi}$ for the transmission
and processing resources, respectively. The heuristic algorithm outperforms
the greedy algorithm by almost 40\%. When the network size is small,
with $|\mathcal{N}|=8$, the performance of the approximation and
greedy algorithms is very close. This is expected since $\varphi$
and $\phi$ are not small compared to the size of the network. Moreover,
the approximation algorithm wastes a potential utilization of $30\%$.

In the second experiment, our aim is to observe the effect of the
incentive for including the set of best-effort NFs ($\eta^{r}$) on
the competitive performance. The experiment is performed over a linear
topology with $20$ nodes. We generate unicast service requests, each
with an overall number of required NFs $(|\mathcal{V}^{r}|)$ of
2, where either one or none of the NFs is set as best-effort in random.
Here, we set $\eta^{r}$ to the number of included NFs, i.e., $\eta_{b}=2$
and $\eta_{m}=1$. Fig. \ref{fig2} shows the normalized aggregate
profit for the three algorithms. As expected, when an incentive is
used for including the best-effort NFs, the aggregate profit is scaled
up for all the algorithms, which implies service requests with best-effort
NFs are encouraged to maximize the aggregate profit. However, this
comes at the expense of an increased competitive ratio compared to
the greedy algorithm. With an incentivized profit function, the percent
increase of the approximation algorithm to the greedy algorithm is
11\%, whereas the percent increase without using an incentive is 39\%.
\begin{figure}
\begin{centering}
\includegraphics[width=0.45\textwidth]{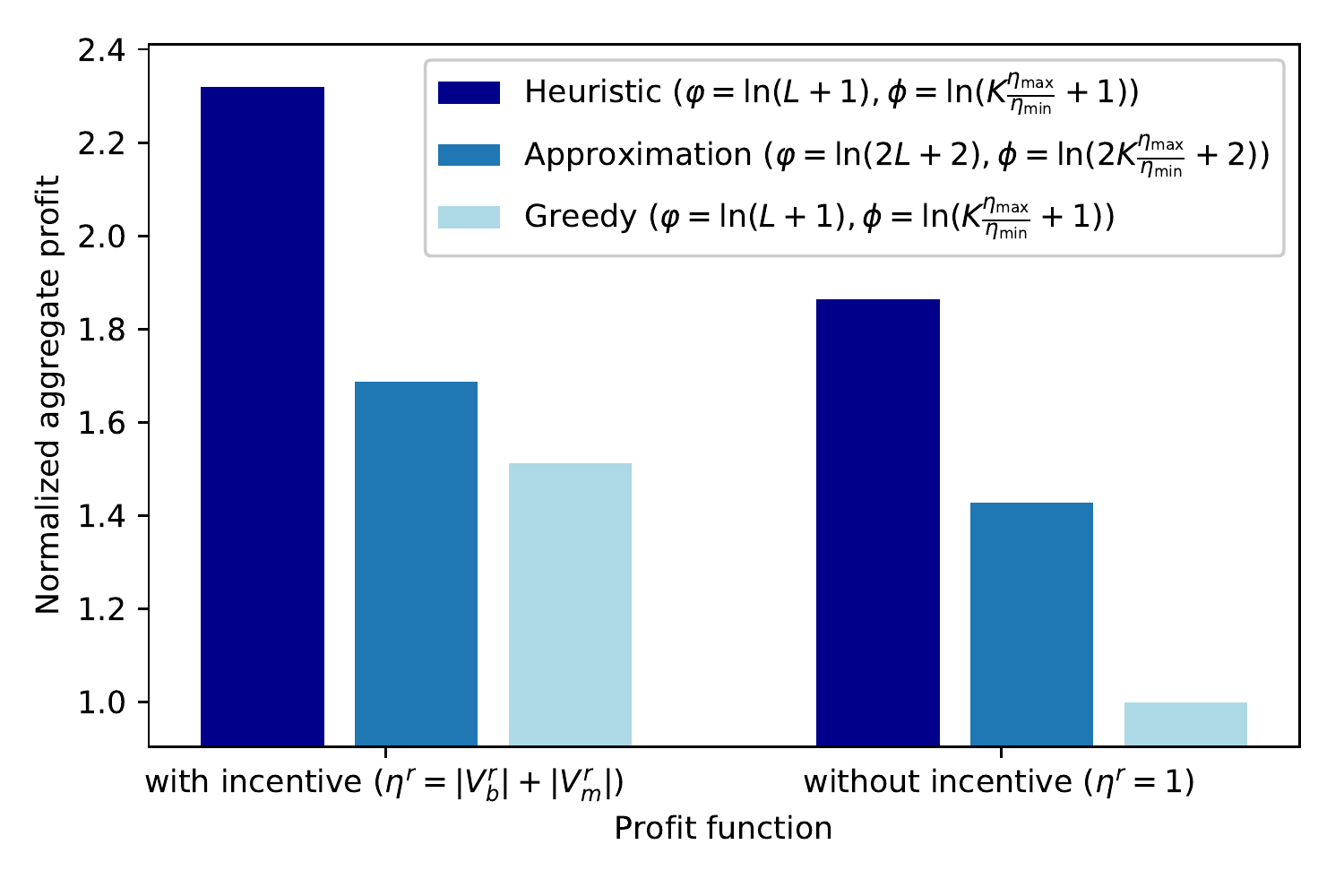}
\par\end{centering}
\caption{Normalized aggregate profit ($\varrho^{r}+\rho^{r}$) for three algorithms
for a linear topology with and without incentivizing the use of best-effort
NFs, with $L=4$ and $K=3$.}
\label{fig2}
\end{figure}

Next, we evaluate the three algorithms on two real topologies from
the Topology Zoo dataset \cite{Topology_zoo}. The first topology,
namely Bell Canada, is a commercial topology with 48 nodes and 64
links. The second topology, namely CESNET, is a research and education
network (REN) with 52 nodes and 63 links. We generate unicast service
requests with 5 required NFs ($|\mathcal{V}^{r}|$). For each service
request, the number of best-effort NFs ($|\mathcal{V}_{b}|$) is uniformly
distributed between 1 and $|\mathcal{V}^{r}|$. In this experiment,
the design goal is to penalize requests which would take unnecessarily
long routes due to deploying NF instances that are far-away from the
shortest path between the source and destination. Therefore, for the
two topologies, we set $L$ to the maximum shortest path between any
pair of nodes, which corresponds to 13 and 6, respectively for each
topology. Moreover, we set $K=5$ and $\eta^{r}=1$. Fig. \ref{fig3}
shows the normalized aggregate profit for the two topologies. The
approximation and greedy algorithms have close performance, whereas
the heuristic algorithm yields a 25\% and 23\% improvement for the
two topologies.
\begin{figure}
\begin{centering}
\includegraphics[width=0.45\textwidth]{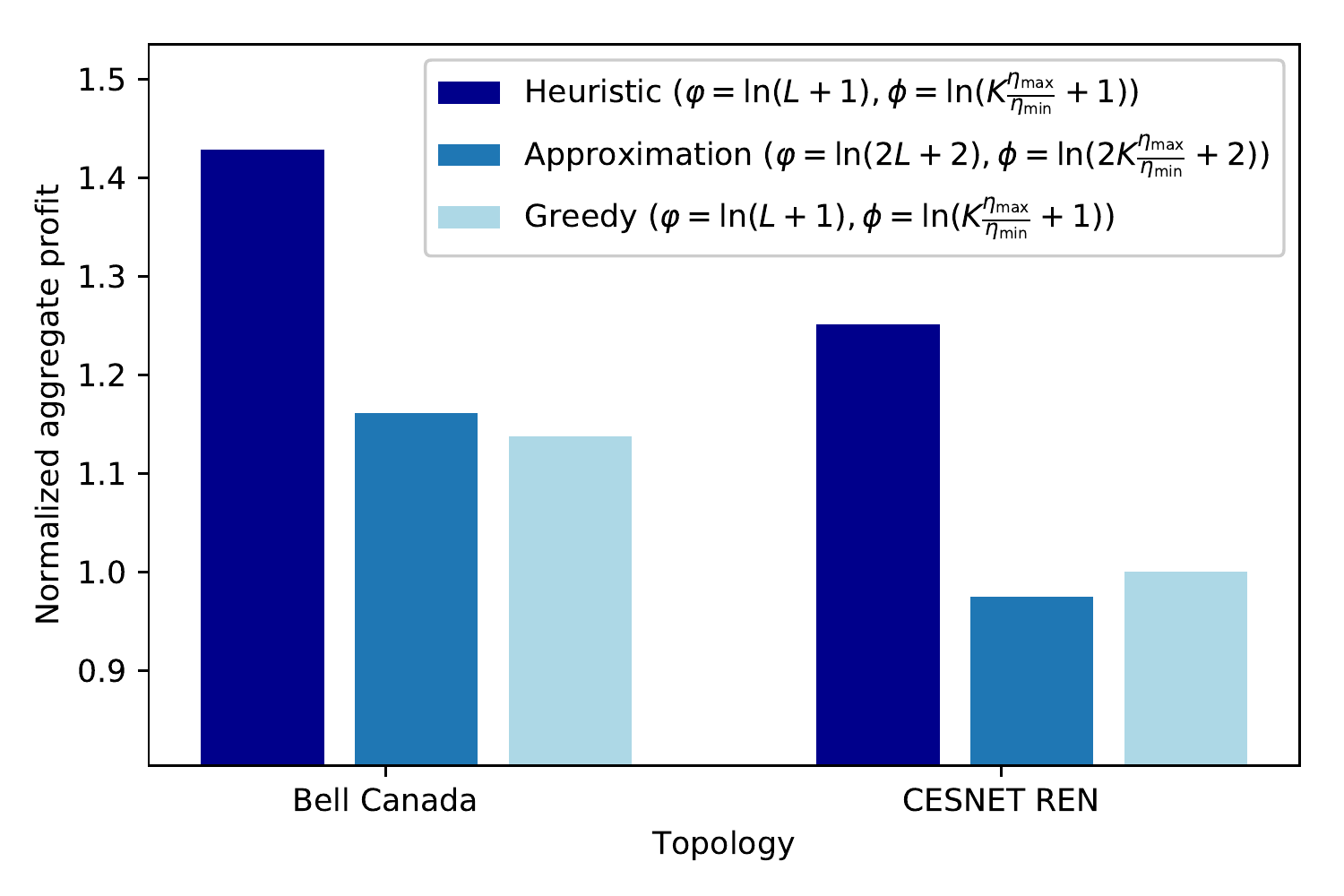}
\par\end{centering}
\caption{Normalized aggregate profit ($\varrho^{r}+\rho^{r}$) for the three
algorithms for two real topologies, namely Bell Canada and CESNET
REN, with $K=5$ and $L=13$ and 6, respectively.}
\label{fig3}
\end{figure}

The next experiment is to test the performance of the online algorithms
on both multicast and unicast service requests over a random topology
with 25 nodes ($|\mathcal{N}|=25$). The random topology is generated
using the Barabási\textendash Albert preferential attachment model,
which provides scale-free network topologies \cite{Barbasi}. For
each service request, the number of destinations varies randomly between
1 and 4, and the number of required NFs in each service request is
uniformly random between 1 and 3. Recall that, to provide a non-discriminatory
treatment between unicast and multicast service requests, $\varrho^{r}\propto|\mathcal{D}|^{k}$,
where $k$ is recommended to be 0.8. Fig. \ref{fig5} shows the normalized
aggregate profit for the online heuristic and greedy algorithms as
$|\mathcal{D}|_{\max}$ grows for different values of $k$. The performance
of the greedy algorithm remains almost constant as $|\mathcal{D}|_{\max}^{k}$
increases. However, the heuristic algorithm shows a downtrend as $|\mathcal{D}|_{\max}^{k}$
increases, especially for large $k$ (e.g, $k=0.8)$. The experiment
demonstrates that the competitive ratio of the online heuristic algorithm
is increased due to the allowed variation in the profit function (as
$|\mathcal{D}|_{\max}^{k}$ increases).

\begin{figure}
\begin{centering}
\includegraphics[width=0.45\textwidth]{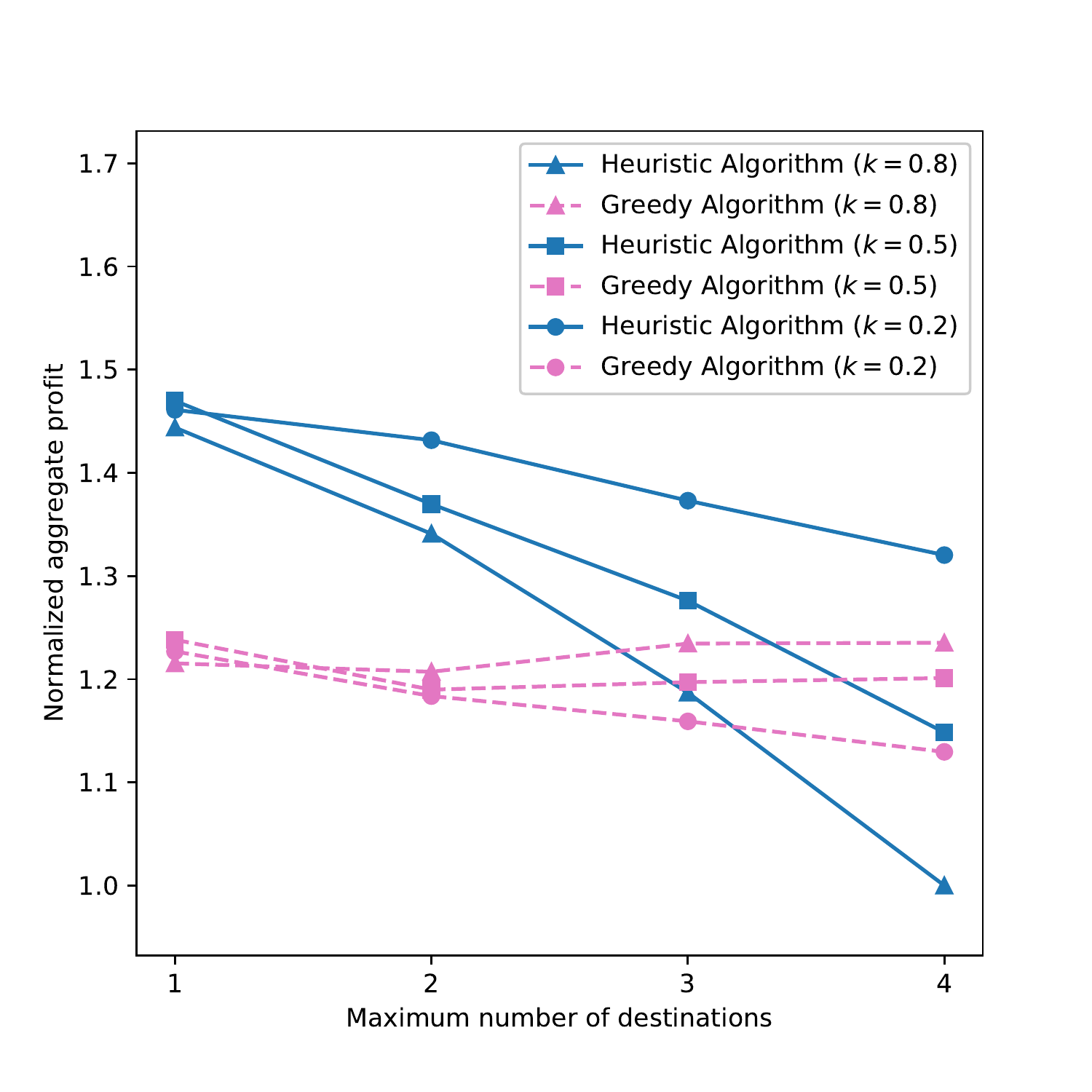}
\par\end{centering}
\caption{Normalized aggregate profit ($\varrho^{r}+\rho^{r}$) for the heuristic
and greedy algorithms over random topology with $|\mathcal{N}|=25$,
$K=4$, $\eta^{r}=1$, and $L$ is set to the maximum hop distance
between any pair of nodes.}
\label{fig5}
\end{figure}

\section{Conclusions}

In this paper, we have proposed a joint admission mechanism and an
online composition, routing and NF placement algorithm for unicast
and multicast NFV-enabled services. We considered services with multiple
mandatory and best-effort NF instances, which is shown to offer a
natural generalization to previous works. Through a primal-dual based
analysis, it is shown that a provable competitive performance can
be achieved, which can be tuned depending on the allowed variability
of the profit function and the desired optimality. This work does
not assume any statistical models on the arrival pattern of the service
requests nor does it have any probabilistic assumptions on the sources,
destinations, and network functions. Therefore, this paper provides
a fundamental understanding on the nature of the profit-maximization
problem for NFV-enabled services with multiple resource types. Indeed,
the analyzed worst-case competitive performance can be improved by
incorporating more contextual assumptions (e.g., adding probabilistic/stochastic
assumptions). In doing so, the performance of the online algorithm
is expected to be improved, while retaining the robustness of the
competitive analysis (at least in a probabilistic/stochastic sense).

\section*{Acknowledgment}

This work was supported by a research grant from the Natural Sciences
and Engineering Research Council (NSERC) of Canada.


\appendix{}

\subsection{Proof for Lemma \ref{lemma:trans}}

\label{sec:append:lemma:trans}

Assume that the transmission resources on physical link $l$ ($\in\mathcal{L}$)
become exceeded when the $r$th service request arrives. Then, we
have $B(l)-\sum_{j=1}^{r-1}d^{j}<d^{r}$. Therefore, after admitting
the $(r-1)$th service, the value of $\bar{x}^{r-1}(l)$ can be expressed
as
\begin{align}
\bar{x}^{r-1}(l) & =\frac{1}{L}(e^{\varphi\frac{\sum_{j=1}^{r-1}d^{j}}{B(l)}}-1)\nonumber \\
 & =\frac{1}{L}(e^{\varphi(1-\frac{B(l)-\sum_{j=1}^{r-1}d^{j}}{B(l)})}-1)\nonumber \\
 & >\frac{1}{L}(e^{\varphi(1-\frac{d^{r}}{B(l)})}-1).\label{eq:lemma1_der_1}
\end{align}
Assuming that the required data rate of a service request is ``small
enough'', i.e. $d^{r}\leq\frac{\min_{l\in\mathcal{L}}B(l)}{\varphi},\,\forall S^{r}\in\sigma$.
Then, inequality (\ref{eq:lemma1_der_1}) becomes
\begin{align}
\bar{x}^{r-1}(l) & \geq\frac{1}{L}(e^{\varphi(1-\frac{1}{\varphi})}-1)\nonumber \\
 & =\frac{1}{L}(\frac{e^{\varphi}}{e}-1).\label{eq:lemma1_ineq1}
\end{align}
Therefore, for the $r$th service request to be rejected, we need
$d^{r}\bar{x}^{r-1}(l)\geq\alpha\varrho^{r},\,\forall S^{r}\in\sigma$,
which translates to $\bar{x}^{r-1}(l)\geq\alpha|\mathcal{D}|_{\max}^{k}.$
Therefore, we need $\frac{1}{L}(\frac{e^{\varphi}}{e}-1)\geq\alpha|\mathcal{D}|_{\max}$,
which entails that $\varphi\geq\ln(2\alpha L|\mathcal{D}|_{\max}^{k}+2)$.
That is, $\varphi$ is set such that the admission mechanism rejects
any service request that would violate the transmission resources
of a physical link.\QEDA

\subsection{Proof for Lemma \ref{lemma_proc}}

\label{sec:append:lemma:proc}

Assuming that the processing resources on the NFV node $n$ ($\in\mathcal{N}$)
become exceeded when request $S^{r}$ is accepted, we have $C(n)-\sum_{j=1}^{r-1}\sum_{f\in\mathcal{V}^{j}}C(f^{j})<C(n).$
Therefore, the value of $\tilde{x}^{r-1}(n)$ can be expressed as
\begin{align}
\tilde{x}^{r-1}(n) & =\frac{1}{K}(e^{\phi\frac{\sum_{j=1}^{r-1}\sum_{f\in\mathcal{V}^{j}}C(f^{j})}{C(n)}}-1)\nonumber \\
 & =\frac{1}{K}(e^{\phi(1-\frac{C(n)-\sum_{j=1}^{r-1}\sum_{f\in\mathcal{V}^{j}}C(f^{j})}{C(n)})}-1)\nonumber \\
 & >\frac{1}{K}(e^{\phi(1-\frac{\sum_{f\in\mathcal{V}^{r}}C(f^{r})}{C(n)})}-1).\label{eq:lemma1_der_1-1}
\end{align}
Under the assumption that the processing requirements of a service
request is ``small enough'', i.e., $C(f^{r})\leq\frac{\min_{n\in\mathcal{N}}C(n)}{\phi},\,S^{r}\in\sigma$,
inequality (\ref{eq:lemma1_der_1-1}) becomes
\begin{align}
\tilde{x}^{r-1}(l) & \geq\frac{1}{K}(e^{\phi(1-\frac{1}{\phi})}-1)\nonumber \\
 & =\frac{1}{K}(\frac{e^{\phi}}{e}-1).\label{eq:lemma1_ineq1-1}
\end{align}
Therefore, for the $r$th service to be rejected, we need $\sum_{n\in\mathcal{N}}C(f^{r})\tilde{x}(n)\geq\beta\rho^{r},\,\forall S^{r}\in\sigma$,
which translates to $\frac{1}{K}(\nicefrac{e^{\phi}}{e}-1)\geq\beta\frac{\eta_{\max}}{\eta_{\min}}$,
i.e., $\phi\geq\ln(2\beta K\frac{\eta_{\max}}{\eta_{\min}}+2)$. That
is, $\phi$ is set such that the admission mechanism rejects any request
that would violate the processing resources of an NFV node.\QEDA

\subsection{Proof for Theorem \ref{proof_compet}}

\label{sec:append:theorem}

Let $\Delta J$ and $\Delta D$ be the change in the primal and dual
cost in each iteration, respectively. Starting with $J=D=0$, when
the $r$th service request is accepted, the objective function of
the dual formulation is increased by $\Delta D=\alpha\varrho^{r}+\beta\rho^{r}$.
The objective function of the primal is increased by
\begin{align}
\nonumber \Delta J & =\sum_{l\in\mathcal{L}}B(l)\big(\bar{x}^{r}(l)-\bar{x}^{r-1}(l)\big)
\\
& + \sum_{n\in\mathcal{N}}C(n)\big(\tilde{x}^{r}(n)-\tilde{x}^{r-1}(n)\big)+z^{r}.
\label{eq:primal_change}
\end{align}
Substituting \eqref{eq:cost_trans} and \eqref{eq:cost_proc} in
\eqref{eq:primal_change}, we obtain
\begin{align}
\Delta J & =\sum_{l\in\mathcal{L}}(e^{\varphi\frac{d^{r}}{B(l)}}-1)(\bar{x}^{r-1}(l)+\nicefrac{1}{L})B(l)\nonumber \\
 & +\sum_{n\in\mathcal{N}}(e^{\phi\frac{C(f^{r})}{C(n)}}-1)(\tilde{x}^{r-1}(n)+\nicefrac{1}{K})C(n)+z^{r}.
\end{align}
Using inequality $e^{x}-1\leq x$ for $0\leq x\leq1$, we get
\begin{align}
\nonumber \Delta J & \leq\varphi\sum_{l\in\mathcal{L}}d^{r}(\bar{x}^{r-1}(l)+\nicefrac{1}{L}) \\
& +\phi\sum_{n\in\mathcal{N}}C(f^{r})(\tilde{x}^{r-1}(n)+\nicefrac{1}{K})+z^{r}.\label{eq:theorem_der3}
\end{align}
From the admission mechanism, substituting $z^{r}$ from \eqref{eq:update_z}
in \eqref{eq:theorem_der3}, we obtain
\begin{align}
\Delta J= & \varphi\sum_{l\in\mathcal{L}}d^{r}(\bar{x}^{r-1}(l)+\nicefrac{1}{L})+\phi\sum_{n\in\mathcal{N}}C(f^{r})(\tilde{x}^{r-1}(n)+\nicefrac{1}{K})\nonumber \\
+ & \max\big(\alpha\varrho^{r}-\sum_{l\in P\cap\mathcal{L}}d^{r}\bar{x}(l),\beta\rho^{r}-\sum_{n\in P\cap\mathcal{N}}C(f^{r})\tilde{x}(n)\big)\label{eq:theorem_der4}\\
\leq & \varphi\sum_{l\in\mathcal{L}}d^{r}(\bar{x}^{r-1}(l)+\frac{\alpha\varrho^{r}}{L})+\phi\sum_{n\in\mathcal{N}}C(f^{r})(\tilde{x}^{r-1}(n)+\frac{\beta\rho^{r}}{K})\nonumber \\
+ & \alpha\varrho^{r}+\beta\rho^{r}-\sum_{l\in P\cap\mathcal{L}}d^{r}\bar{x}(l)-\sum_{n\in P\cap\mathcal{N}}C(f^{r})\tilde{x}(n)\label{eq:theorem_der5}\\
= & (\varphi-1)\sum_{l\in\mathcal{L}}d^{r}\bar{x}^{r-1}(l)+\alpha\varrho^{r}+\beta\rho^{r}+(\phi-1)\nonumber \\
\times & \sum_{n\in\mathcal{N}}C(f^{r})\tilde{x}^{r-1}(n)+\alpha\varphi\varrho^{r}\sum_{l\in\mathcal{L}}\frac{d^{r}}{L}+\beta\phi\rho^{r}\frac{1}{K}\sum_{n\in\mathcal{N}}C(f^{r})\label{eq:theorem_der7}\\
\leq & \alpha\varrho^{r}+\beta\rho^{r}+(\varphi-1)\sum_{l\in\mathcal{L}}d^{r}\bar{x}^{r-1}(l)\nonumber \\
+ & (\phi-1)\sum_{n\in\mathcal{N}}C(f^{r})\tilde{x}^{r-1}(n)+\alpha\varrho^{r}\varphi+\beta\rho^{r}\phi.\label{eq:theorem_der8}
\end{align}
Since the $r$th service request is accepted, i.e., $\sum_{l\in\mathcal{L}}d^{r}\bar{x}^{r-1}(l)\leq\alpha\varrho^{r}$
and $\sum_{n\in\mathcal{N}}C(f^{r})\tilde{x}^{r-1}(n)\leq\beta\rho^{r}$,
inequality \eqref{eq:theorem_der8} becomes
\begin{align}
\Delta J\leq & 2\alpha\varrho^{r}\varphi+2\beta\rho^{r}\phi\label{eq:theorem_der10}\\
\leq & 2(\alpha\varrho^{r}+\beta\rho^{r})\max\{\varphi,\phi\}.\label{eq:theorem_DeltaP}
\end{align}
It is shown in Lemmata \ref{lemma:trans} and \ref{lemma_proc} that
the online algorithm ensures that the transmission and processing
resource constraints are always satisfied, where variables $\bar{x}(l),\,\tilde{x}(n)$,
and $z^{r}$ are designed such that a feasible primal solution is
maintained. Therefore, using weak duality (i.e., $\Delta D\leq\Delta J$)
and from \eqref{eq:theorem_DeltaP}, a competitive performance of
$\mathcal{O}\big(\max\{\varphi,\phi\}\big)$ is concluded.\QEDA

\bibliographystyle{myIEEEtran}
\bibliography{references}

\begin{IEEEbiography}[{\includegraphics[width=1in,height=1.25in,clip,keepaspectratio]{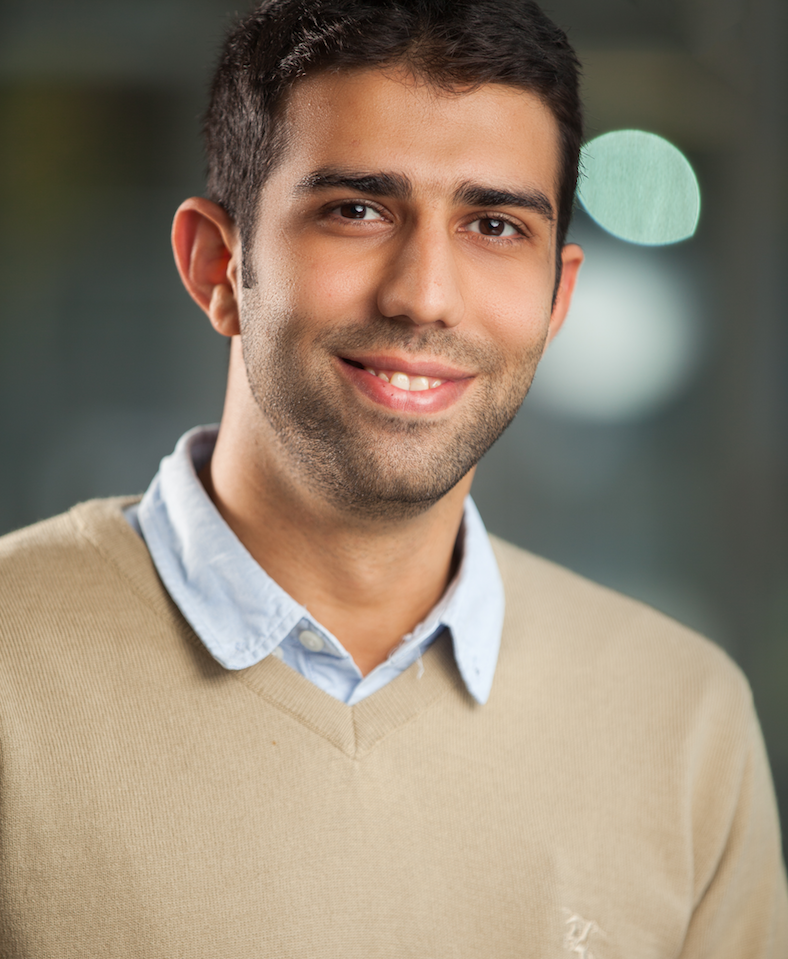}}]{Omar Alhussein}
is currently persuing the Ph.D. degree in electrical engineering with the University of Waterloo, Ontario, Canada. He received the M.A.Sc. degree in engineering science from Simon Fraser University, British Columbia, Canada, in 2015, and the B.Sc. degree in communications engineering from Khalifa University, Abu Dhabi, United Arab Emirates, in 2013. His research interests include next generation wireless networks, network (function) virtualization, wireless communications, and machine learning.
\end{IEEEbiography}

\begin{IEEEbiography}[{\includegraphics[width=1in,height=1.25in,clip,keepaspectratio]{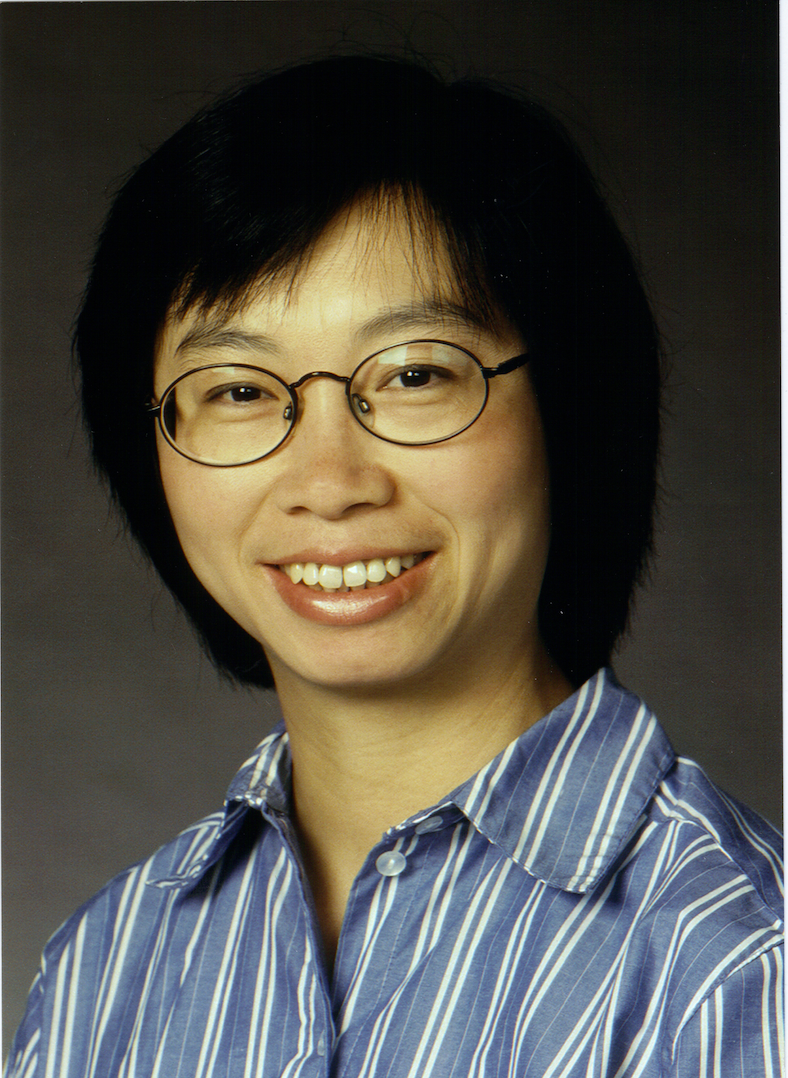}}]{Weihua Zhuang}
(M'93$-$SM'01$-$F'08) has been with the Department of Electrical and Computer Engineering, University of Waterloo, Canada, since 1993, where she is a Professor and a Tier I Canada Research Chair in wireless communication networks. She was a recipient of the 2017 Technical Recognition Award from the IEEE Communications Society Ad Hoc and Sensor Networks Technical Committee, and several best paper awards from IEEE conferences. She was the Technical Program Chair/Co-Chair of the IEEE VTC 2016 Fall and 2017 Fall, the Editor-in-Chief of IEEE Transactions on Vehicular Technology (2007$-$2013), and an IEEE Communications Society Distinguished Lecturer (2008$-$2011). Dr. Zhuang is a Fellow of the Royal Society of Canada, the Canadian Academy of Engineering, and the Engineering Institute of Canada. She is an Elected Member of the Board of Governors and VP Publications of the IEEE Vehicular Technology Society.
\end{IEEEbiography}

\end{document}